\documentclass[reqno,11pt]{amsart}
\usepackage{amsmath,amssymb,amsthm,graphicx,a4wide,enumerate,url}
\usepackage[small,bf]{caption}
\theoremstyle{plain}
\newtheorem{thm}{Theorem}
\newtheorem{lem}[thm]{Lemma}

\theoremstyle{definition}

\theoremstyle{remark}
\newtheorem{rem}{Remark}

\newcommand{\prn}[1]{\left(#1\right)}

\newcommand{\pd}[2]{\frac{\partial#1}{\partial#2}}
\newcommand{\ud}[1]{\,\mathrm{d}#1}
\renewcommand{\vec}[1]{\mathbf{#1}}

\begin{document}
\parskip.9ex

\title[Comparative Accuracy of a Data-Fitted Generalized Aw-Rascle-Zhang Model]
{Comparative Model Accuracy of a Data-Fitted Generalized Aw-Rascle-Zhang Model}
\author[S. Fan]{Shimao Fan}
\address[Shimao Fan]
{Department of Civil and Environmental Engineering \\ \newline
University of Illinois at Urbana Champaign \\
205 N. Mathews Ave \\ Urbana, IL 61801}
\email{shimao@illinois.edu}
\author[M. Herty]{Michael Herty}
\address[Michael Herty]
{Department of Mathematics \\ RWTH Aachen University \\ \newline
Templergraben 55 \\ D-52056 Aachen}
\email{herty@mathc.rwth-aachen.de}
\urladdr{http://www.opt.rwth-aachen.de/herty}
\author[B. Seibold]{Benjamin Seibold}
\address[Benjamin Seibold]
{Department of Mathematics \\ Temple University \\ \newline
1805 North Broad Street \\ Philadelphia, PA 19122}
\email{seibold@temple.edu}
\urladdr{http://www.math.temple.edu/\~{}seibold}

\begin{abstract}
The Aw-Rascle-Zhang (ARZ) model can be interpreted as a generalization of the Lighthill-Whitham-Richards (LWR) model, possessing a family of fundamental diagram curves, each of which represents a class of drivers with a different empty road velocity. A  weakness of this approach is that different drivers possess vastly different densities at which traffic flow stagnates. This drawback can be overcome by modifying the pressure relation in the ARZ model, leading to the generalized Aw-Rascle-Zhang (GARZ) model. We present an approach to determine the parameter functions of the GARZ model from fundamental diagram measurement data. The predictive accuracy of the resulting data-fitted GARZ model is compared to other traffic models by means of a three-detector test setup, employing two types of data: vehicle trajectory data, and sensor data. This work also considers the extension of the ARZ and the GARZ models to models with a relaxation term, and conducts an investigation of the optimal relaxation time.
\end{abstract}

\maketitle

\section{Introduction}
The mathematical modeling of vehicular traffic flow knows a variety of types of descriptions (see \cite{Helbing2001, BellomoDogbe2011} for review papers): \emph{microscopic} (e.g., \cite{Pipes1953, Newell1961, BandoHesebeNakayama1995, Helbing1995}), which model the individual vehicles and their interactions by ODE; \emph{cellular} (e.g., \cite{NagelSchreckenberg1992, FukuiIshibashi1996, Daganzo2006, SakaiNishinariIIda2006, AlperovichSopasakis2008}), which divide the road into cells and prescribe stochastic rules how vehicles advance through cells; and \emph{continuum}. This latter class divides into \emph{mesoscopic/gas-kinetic} (e.g., \cite{HermanPrigogine1971, Phillips1979, KlarWegener2000, IllnerKlarMaterne2003, HertyPareschi2010, HertyIllner2010}) and \emph{macroscopic} (e.g., \cite{LighthillWhitham1955, Richards1956, Underwood1961, Payne1971, Payne1979, Lebacque1993, KernerKonhauser1993, KernerKonhauser1994, Daganzo1995, AwRascle2000, Daganzo2006, GaravelloPiccoli2006, Goatin2006, BerthelinDegondDelitalaRascle2008, BayenClaudel2011}), i.e., fluid-dynamical, models. Among the (inviscid) macroscopic models one distinguishes between first-order models based on scalar hyperbolic equations and second-order models comprised of systems of hyperbolic equations. Specific examples for the latter are the Payne-Whitham model \cite{Payne1971, Whitham1974}, two-phase models \cite{ColomboGoatin2006}, and the Aw-Rascle-Zhang model \cite{AwRascle2000, Greenberg2001, Zhang2002, BerthelinDegondDelitalaRascle2008}.

All of these types of traffic models are of practical importance, however---due to their different mathematical structure---for different purposes. For instance, microscopic models are well-suited for traffic simulation, i.e., the ``in silico'' study of a specific scenario; cellular models reproduce jamming behavior while being simple to implement and easy to parallelize; mesoscopic models provide a high modeling flexibility; and macroscopic models provide a suitable framework for the incorporation of on-line data. Moreover, there are mathematical relations between these types of models. For example, microscopic models as well as cellular models converge to mesoscopic or macroscopic models in the limit of vanishing cell size or vehicle spacing, respectively \cite{Daganzo1994, Daganzo2006, AlperovichSopasakis2008, BorscheKimathiKlar2012}. Similarly, macroscopic models arise as suitable limits of mesoscopic models \cite{NelsonSopasakis1999, KlarWegener2000}.

Microscopic and cellular models are nowadays widely used in traffic engineering, and their combination with data is ubiquitous. In contrast, certain types of continuum models have been studied mathematically, but very little work has been conducted on their validation with traffic data. Examples of macroscopic first-order models used in traffic engineering practice are the Mobile Century project \cite{Aminetal2008} and the Mobile Millennium project \cite{MobileMillennium}, including approaches based on the reformulation of the Lighthill-Whitham-Richards model \eqref{eq:lighthill_whitham_richards_model} in terms of Hamilton-Jacobi equations \cite{BayenClaudel2010, BayenClaudel2011}, or in terms of the velocity variable \cite{WorkBlandinTossavainenPiccoliBayen2010}. Further examples are \cite{BlandinBrettiCutoloPiccoli2009, BlandinCoqueBayen2012} and the references therein. Those projects focus on the assimilation of data, the reconstruction of traffic states (e.g., from cell phone data), and the combination of macroscopic models with filtering techniques.

In contrast, here different types of traffic models are generated via historic data, and then their predictive accuracy is investigated using time-dependent data. Moreover, the main focus lies on second-order models; and first-order models are considered mainly for comparison purposes. The contributions of this paper are: (i) the design of a generalized Aw-Rascle-Zhang model and the analysis of its mathematical properties; (ii) a systematic methodology to construct data-fitted first-order and second-order traffic models, using historic fundamental diagram data; (iii) the validation of first- and second-order macroscopic traffic models via time-dependent trajectory and sensor data, and the comparison of the predictive accuracy of different models; and (iv) the investigation of the optimal relaxation time in second-order models with a relaxation term.

This paper is organized as follows. In \S\ref{sec:models} an overview over existing macroscopic traffic models is provided, including a discussion of some of the modeling shortcomings of the Aw-Rascle-Zhang model. We then introduce the generalized Aw-Rascle-Zhang model as an approach that addresses the shortcomings, and discuss its mathematical properties. The fitting of the model parameters and functions is then described in \S\ref{sec:data-fitted_models}. Given historic fundamental diagram data in the flow rate vs.\ density plane, we systematically construct data-fitted first- and second-order macroscopic models. In \S\ref{sec:numerical_methods} the numerical methods used to conduct the model validation and comparison are presented. Unlike studies of cell-transmission models \cite{Daganzo1994}, in this paper all studies are carried out in a macroscopic sense; in particular the governing PDE are numerically solved with high enough accuracy such that the numerical approximation errors are negligibly small relative to the model errors. The comparison of the models on a three-detector test setup \cite{Daganzo1997} is then carried out in \S\ref{sec:validation}. In addition to the macroscopic traffic models, we also consider a predictor that simply interpolates the traffic state from the boundaries. For vehicle trajectory data, and for sensor data, the predictive accuracies of the models are compared, and their reproduction of features in the traffic states are studied. In \S\ref{sec:inhomogeneous_models} we then extend the studies to data-fitted second-order models with relaxation terms. In particular, we study the dependence of the model accuracy on the relaxation time at which drivers adjust their driving behavior. Finally, in \S\ref{sec:conclusions} we present the conclusions from our studies.

\vspace{1.5em}
\section{Existing and New Macroscopic Traffic Models}
\label{sec:models}
Common to all macroscopic traffic models is the continuity equation
\begin{equation}
\label{eq:continuity_equation}
\rho_t+(\rho u)_x = 0\;,
\end{equation}
which gives the conservation of vehicles. In \eqref{eq:continuity_equation}, the vehicle density is $\rho(x,t)$, and the vehicle velocity field is $u(x,t)$, where $x$ is the position along the road, and $t$ is time. If the road has multiple lanes (in a given direction), we consider these aggregated into the scalar field quantities $\rho$ and $u$.

\subsection{The Lighthill-Whitham-Richards Model and Flow Rate Functions}
The simplest macroscopic traffic model, the Lighthill-Whitham-Richards (LWR) model \cite{LighthillWhitham1955, Richards1956}, is obtained by assuming a functional relationship between $\rho$ and $u$, i.e., $u = U(\rho)$. This turns equation \eqref{eq:continuity_equation} into a scalar hyperbolic conservation law
\begin{equation}
\label{eq:lighthill_whitham_richards_model}
\rho_t+(Q(\rho))_x = 0\;,
\end{equation}
where the flux $Q$ is given by the flow rate function $Q(\rho) = \rho U(\rho)$. Because the LWR model \eqref{eq:lighthill_whitham_richards_model} is a closed model consisting of a single equation, it is denoted a \emph{first order model}. The velocity function $U(\rho)$ is commonly assumed to be decreasing in $\rho$ with $U(\rho_\text{max}) = 0$ for some maximal vehicle density $\rho_\text{max}>0$.

Popular examples of flow rate functions are the Greenshields flux \cite{Greenshields1935}, in which $Q(\rho)$ is a quadratic function, and the Newell-Daganzo flux \cite{Newell1993, Daganzo1994}, in which $Q(\rho)$ is a piecewise linear function. While these different choices of functions $Q(\rho)$ lead to well-posed first-order models, the second-order models derived below call for further properties that the function $Q(\rho)$ must satisfy. In particular, the velocity function $U(\rho) = Q(\rho)/\rho$ must nowhere be constant, because otherwise hyperbolicity would be lost (see the analysis in \S\ref{subsubsec:characteristics}). This rules out the Newell-Daganzo flux; and as a consequence, in this paper we consider flow rate functions \eqref{eq:flow_rate_curve} that resemble the shape of the Newell-Daganzo flux, but that are strictly concave.

\subsection{The Aw-Rascle-Zhang Model}
The strict functional relationship between $\rho$ and $u$ is loosened in \emph{second order models}, which augment \eqref{eq:continuity_equation} by an evolution equation for the velocity field. Payne and Whitham proposed a model \cite{Payne1971, Whitham1974} in which the vehicle velocity relaxes towards a velocity function, while also being affected by a ``traffic pressure'' that is analogous to a pressure in fluid dynamics. Because this pressure can lead to unrealistic solutions (vehicles going backwards on the road, shocks that overtake vehicles, etc., see \cite{Daganzo1995}), Aw and Rascle \cite{AwRascle2000}, and independently Zhang \cite{Zhang2002}, proposed a different form of ``pressure'' that remedies the shortcomings of the Payne-Whitham model. The homogeneous Aw-Rascle-Zhang (ARZ) model reads as
\begin{equation}
\label{eq:aw_rascle_zhang_model_homogeneous}
\begin{split}
\rho_t+(\rho u)_x &= 0\;, \\
(u+h(\rho))_t+u(u+h(\rho))_x &= 0\;,
\end{split}
\end{equation}
where we call $h(\rho)$ the \emph{hesitation function}.\footnote{The function $h(\rho)$ is sometimes called ``pressure'', and denoted $p(\rho)$, even though it does not play the role of a pressure in the equations.} We assume that $h'(\rho)>0$ and use the gauge $h(0) = 0$. The addition of a relaxation term (analogous to the Payne-Whitham model) yields the inhomogeneous ARZ model \cite{Greenberg2001, Rascle2002}
\begin{equation}
\label{eq:aw_rascle_zhang_model_inhomogeneous}
\begin{split}
\rho_t+(\rho u)_x &= 0\;, \\
(u+h(\rho))_t+u(u+h(\rho))_x &= \tfrac{1}{\tau}\prn{U_\text{eq}(\rho)-u}\;.
\end{split}
\end{equation}
We call $U_\text{eq}(\rho)$ the \emph{desired velocity function} or the \emph{equilibrium velocity function}, and $\tau$ the \emph{relaxation time scale}.

As one can easily verify, the homogeneous ARZ model \eqref{eq:aw_rascle_zhang_model_homogeneous} possesses no mechanism to make drivers move when starting with all vehicles at rest, i.e., $u(x,0) = 0$. In turn, the inhomogeneous ARZ model \eqref{eq:aw_rascle_zhang_model_inhomogeneous} does. We therefore expect the homogeneous ARZ model to yields reasonable results only when the traffic flow is close to its equilibrium state, i.e., $u\approx U_\text{eq}(\rho)$. Yet, in general the inhomogeneous ARZ model has the potential to yield more realistic predictions.

\begin{rem}
The conservative form of \eqref{eq:aw_rascle_zhang_model_inhomogeneous} is given by
\begin{equation}
\label{eq:aw_rascle_zhang_model_conservative}
\begin{split}
\rho_t+\prn{q-\rho h(\rho)}_x &= 0\;, \\
q_t+\prn{\tfrac{q^2}{\rho}-h(\rho)q}_x &= \tfrac{1}{\tau}\prn{Q_\text{eq}(\rho)+\rho h(\rho)-q}\;,
\end{split}
\end{equation}
where the two conserved variables are $\rho$ and $q = \rho (u+h(\rho))$, and $Q_\text{eq}(\rho) = \rho U_\text{eq}(\rho)$ is called the \emph{equilibrium curve} that the momentum density $\rho u$ relaxes to.
\end{rem}

\begin{rem}
\label{rem:dependent_h_U}
Various authors (e.g., \cite{Greenberg2001, SiebelMauser2006}) have proposed to choose the functions $h(\rho)$ and $U_\text{eq}(\rho)$ dependent on each other, namely $h(\rho) = U_\text{eq}(0)-U_\text{eq}(\rho)$. In this case, the solution relaxes towards an equilibrium state $u = U_\text{eq}(\rho)$ (see \S\ref{subsubsec:relaxtion_GARZ_LWR}). In this paper, we follow this philosophy, since it generates both functions from the same data-fitting procedure. However, it should be noted that it is in principle perfectly reasonable to choose $h(\rho)$ and $U_\text{eq}(\rho)$ independently of each other. As shown in \cite{Greenberg2004, FlynnKasimovNaveRosalesSeibold2009, SeiboldFlynnKasimovRosales2013, KasimovRosalesSeiboldFlynn2013}, such a choice can generate (whenever $h'(\rho)+U_\text{eq}'(\rho)<0$) instabilities and self-sustaining traveling wave solutions that model phantom traffic jams and traffic waves, respectively.
\end{rem}

\subsection{Interpretation of ARZ as Generalization of LWR}
\label{subsec:ARZ_generalization_of_LWR}
As pointed out in \cite{Lebacque1993, AwRascle2000, BerthelinDegondDelitalaRascle2008, FanSeibold2013}, the homogeneous ARZ model \eqref{eq:aw_rascle_zhang_model_homogeneous} can be interpreted as a generalization of the LWR model, by introducing the variable $w = u+h(\rho)$. Thus, system \eqref{eq:aw_rascle_zhang_model_homogeneous} takes the form
\begin{equation}
\label{eq:aw_rascle_zhang_model_w}
\begin{split}
\rho_t+(\rho u)_x &= 0\;, \\
w_t+uw_x &= 0\;, \\
\text{where~}u &= w-h(\rho)\;.
\end{split}
\end{equation}
The interpretation of \eqref{eq:aw_rascle_zhang_model_w} is that $w$ is advected with the flow $u$, i.e., it moves with the vehicles. One can therefore interpret $w$ as a quantity that is associated with each vehicle, and that is influencing the velocity. Since for $\rho = 0$, we have $h(0) = 0$ and thus $w = u$, we call $w$ the \emph{empty road velocity}. The actual vehicle velocity is given by its empty road velocity, reduced by the hesitation $h(\rho)$.

Using this interpretation, the homogeneous ARZ model generalizes the LWR model \eqref{eq:lighthill_whitham_richards_model}, as follows. Given a (decreasing) LWR velocity function $U(\rho)$, we define $h(\rho) = U(0)-U(\rho)$ (clearly, $h'(\rho)>0$ and $h(0) = 0$). Then, model \eqref{eq:aw_rascle_zhang_model_w} possesses a one-parameter family of velocity curves, namely $u_w(\rho) = w-h(\rho) = U(\rho)+(w-U(0))$, and the LWR velocity curve $u(\rho) = U(\rho)$ is one of them, namely the one corresponding to $w = U(0)$. The same behavior translates to the flow rate curves that live in the fundamental diagram ($Q$~vs.~$\rho$). The ARZ model \eqref{eq:aw_rascle_zhang_model_w} possesses a one-parameter family of flow rate curves, namely $Q_w(\rho) = Q(\rho)+\rho(w-U(0))$, and the LWR flow rate curve $Q(\rho)$ is one of them, namely the one corresponding to $w = U(0)$. This has been observed by Lebacque \cite{Lebacque1993}.

The aforementioned relationship between the LWR and the ARZ model is shown in Fig.~\ref{fig:fd_lwr_arz}. The single velocity curve (left panel) and flow rate curve (right panel), shown in red, is one representative of the family of curves, shown in black, that the ARZ model possesses. By construction, each velocity curve in the ARZ model is merely a vertical translation of the LWR velocity curve. Hence, the homogeneous ARZ model possesses the same number of parameter functions (namely: a single one) as the LWR model.

In line with Remark~\ref{rem:dependent_h_U}, i.e., by choosing $U_\text{eq}(\rho) = U(\rho) = U(0)-h(\rho)$, we can extend model \eqref{eq:aw_rascle_zhang_model_w} to the inhomogeneous case, yielding
\begin{equation}
\label{eq:aw_rascle_zhang_model_w_inhomogeneous}
\begin{split}
\rho_t+(\rho u)_x &= 0\;, \\
w_t+uw_x &= \tfrac{1}{\tau}\prn{U(0)-w}\;, \\
\text{where~}u &= w-h(\rho)\;,
\end{split}
\end{equation}
which adds a temporal relaxation of each vehicle's empty road velocity $w$ towards a uniform value $U(0)$. In other words: the dynamics, that can be on any velocity curve $u_w(\rho)$, are driven towards the LWR velocity function $U(\rho)$, i.e., the red curves in Fig.~\ref{fig:fd_lwr_arz}---unless the system is driven away from equilibrium by another effect, such as boundary conditions (see \S\ref{subsubsec:relaxtion_GARZ_LWR}).

\subsection{Generalized ARZ Model}
\label{subsec:garz}
The interpretation of the ARZ model \eqref{eq:aw_rascle_zhang_model_homogeneous} as possessing a family of velocity curves (see \eqref{eq:aw_rascle_zhang_model_w}) reveals a fundamental shortcoming of the model: due to the additive relationship between velocity, empty road velocity, and hesitation, there is not a unique maximum density $\rho_\text{max}$, at which the flow stagnates. On the contrary, as the plots in Fig.~\ref{fig:fd_lwr_arz} indicate, variations in $w$ can lead to significant variations in the density at which $u_w(\rho) = 0$. However, since in reality the maximum density is largely a property of the road, it should not depend (at least not strongly) on the velocity that drivers assume when alone on the road. In order to remedy this shortcoming, the relationship between $u$, $w$, and $\rho$ must be generalized.

To that end, we consider a generalized Aw-Rascle-Zhang (GARZ) model, which is a representative of the class of generic second order models (GSOM), proposed by Lebacque, Mammar, and Haj-Salem \cite{LebacqueMammarHajSalem2007}. Specifically, the homogeneous ARZ model \eqref{eq:aw_rascle_zhang_model_w} generalizes to
\begin{equation}
\label{eq:GARZ_model}
\begin{split}
\rho_t+(\rho u)_x &= 0\;, \\
w_t+uw_x &= 0\;, \\
\text{where~}u &= V(\rho,w)\;,
\end{split}
\end{equation}
where we impose the following requirements on the velocity function $V(\rho,w)$, and the associated generalized flow rate function $Q(\rho,w) = \rho V(\rho,w)$:
\begin{itemize}
\item $V(\rho,w)\ge 0$, i.e., vehicles never go backwards on the road.
\item $V(0,w) = w$, i.e., we gauge the convected quantity $w$ to play the role of the empty road velocity, as in the ARZ model \eqref{eq:aw_rascle_zhang_model_w}.
\item $\frac{\partial^2 Q}{\partial \rho^2}(\rho,w) < 0$ for $w > 0$, i.e., each flow-rate curve $Q_w(\rho) = Q(\rho,w)$ is strictly concave. This condition implies (see Lemma~\ref{lem:concave_Q_decreasing_U}) in particular that $\pd{V}{\rho}(\rho,w) < 0$, i.e., each velocity curve $u_w(\rho) = V(\rho,w)$ is strictly decreasing w.r.t.\ the density.
\item $\pd{V}{w}(\rho,w) > 0$, i.e., a faster empty road velocity results in a faster velocity for all possible densities.
\item $V(\rho,0) = 0$, i.e., for $w = 0$, the concavity of $Q$ and the slope of $V$ hold with an equality sign.
\end{itemize}
\begin{lem}
\label{lem:concave_Q_decreasing_U}
Consider a $C^2$ function $U(\rho)$, and let $Q(\rho) = \rho U(\rho)$. If $Q''(\rho)<0$ everywhere, then $U'(\rho)<0$ everywhere.
\end{lem}
\begin{proof}
The function $a(\rho) = Q(\rho)-Q'(\rho)\rho$ satisfies: (i) $a(0) = 0$, and (ii) $a'(\rho) = -Q''(\rho)\rho > 0$ everywhere. Hence, $U'(\rho) = \frac{-a(\rho)}{\rho^2} < 0$ everywhere.
\end{proof}

In order to define an inhomogeneous GARZ model, an equilibrium velocity curve must be specified. We assume that it is a member of the family of velocity curves defined by $V$, i.e.
\begin{equation*}
U_\text{eq}(\rho) = V(\rho,w_\text{eq})\;,
\end{equation*}
for some equilibrium empty road velocity $w_\text{eq}$. We choose to generalize \eqref{eq:aw_rascle_zhang_model_inhomogeneous} to the GARZ case as follows:
\begin{equation}
\label{eq:GARZ_model_inhomogeneous}
\begin{split}
\rho_t+(\rho u)_x &= 0\;, \\
w_t+uw_x &= \tfrac{1}{\tau}\prn{U_\text{eq}(\rho)-u}\;, \\
\text{where~}u &= V(\rho,w)\;.
\end{split}
\end{equation}
Note that it would alternatively be conceivable to propose a relaxation in \eqref{eq:GARZ_model_inhomogeneous} of the form
\begin{equation}
\label{eq:alternative_relaxation}
w_t+uw_x = \tfrac{1}{\tau}\prn{w_\text{eq}-w}\;.
\end{equation}
While for the ARZ model, the forms \eqref{eq:aw_rascle_zhang_model_inhomogeneous} and \eqref{eq:aw_rascle_zhang_model_w_inhomogeneous} are equivalent, for the GARZ model the relaxations \eqref{eq:GARZ_model_inhomogeneous} and \eqref{eq:alternative_relaxation} are not. Specifically, if $\pd{V}{w}(\rho,w_\text{eq})>0$, then both forms relax to the same limit, but at different rates. A simple Taylor expansion yields that for $w$ nearby $w_\text{eq}$, the relaxation in \eqref{eq:GARZ_model_inhomogeneous} happens $\pd{V}{w}(\rho,w_\text{eq})$ times as rapidly as the relaxation in \eqref{eq:alternative_relaxation}.

Finally, in line with \eqref{eq:aw_rascle_zhang_model_conservative}, the GARZ model \eqref{eq:GARZ_model_inhomogeneous} is meant to be interpreted in the conservative form
\begin{equation}
\label{eq:GARZ_conservative}
\begin{split}
\rho_t+\prn{V(\rho,q/\rho)\rho}_x &= 0\;, \\
q_t+\prn{V(\rho,q/\rho)q}_x &= \tfrac{1}{\tau}\prn{Q_\text{eq}(\rho)-Q(\rho,q/\rho)}\;,
\end{split}
\end{equation}
where the two conserved variables are $\rho$ and $q = \rho w$. Moreover, $Q_\text{eq}(\rho) = \rho U_\text{eq}(\rho)$ and $Q(\rho,w) = \rho V(\rho,w)$.

\subsection{Properties of the GARZ Model}
Most properties of the classical ARZ model transfer over to its generalization, the GARZ model. Here we only collect relevant results, many of which have been presented in \cite{LebacqueMammarHajSalem2007}, or that are relatively straightforward generalizations of the results given in \cite{AwRascle2000, Greenberg2001, Rascle2002}. The theoretical results presented below are in particular important for the data-fitting methodologies conducted in \S\ref{sec:data-fitted_models}, and for the interpretation of the results obtained in \S\ref{sec:inhomogeneous_models}.

\subsubsection{Regions of GARZ Variables and Inverse Velocity Functions}
\label{subsubsec:regions}
Because of their relations and because of their physical meaning, the quantities $\rho$, $w$, and $u$ cannot assume any arbitrary values. In this paper, we assume that there is a unique stagnation density $\rho_\text{max}$, at which vehicles come to a stop, independent of their empty road velocity, i.e., $V(\rho_\text{max},w) = 0$ for all $w$. We therefore have $\rho\in [0,\rho_\text{max})$ and $0 < u\le w$, where the latter inequality follows from $\pd{V}{\rho}(\rho,w) < 0$. Moreover, we assume that there is a minimum and maximum empty road velocity, i.e., $0<w_\text{min}\le w\le w_\text{max}$. Because $\pd{V}{\rho} < 0$ and $\pd{V}{w} > 0$, the function $V(\rho,w)$ can be ``inverted'' to define the functions (and their domains):
\begin{align*}
V:\mathcal{D}_V \longrightarrow [0,w_\text{max}]
&\;\,\text{where~~}
\mathcal{D}_V = \{(\rho,w)\;|\; 0\le\rho<\rho_\text{max},\;
w_\text{min}\le w\le w_\text{max}\}\;, \\
R:\mathcal{D}_R \longrightarrow [0,\rho_\text{max})
&\;\,\text{where~~}
\mathcal{D}_R = \{(u,w)\;|\; 0<u\le w,\; w_\text{min}\le w\le w_\text{max}\}\;, \\
W:\mathcal{D}_W\! \longrightarrow [w_\text{min},w_\text{max}]
&\;\,\text{where~~}
\mathcal{D}_W\! = \{(\rho,u)\;|\; 0\le\rho<\rho_\text{max},\,
V(\rho,w_\text{min})\le u\le V(\rho,w_\text{max})\}\;.
\end{align*}
Here the functions $R$ and $W$ are defined as the unique solutions to the problems:
\begin{enumerate}[\quad a)]
\item given $u$ and $w$, find $\rho = R(u,w)$, s.t.~$V(\rho,w) = u$;
\item given $\rho$ and $u$, find $w = W(\rho,u)$, s.t.~$V(\rho,w) = u$.
\end{enumerate}
From the fact that the quantity $w$ is transported with the flow (while possibly relaxing to some $w_\text{eq} \in (w_\text{min},w_\text{max})$), and from the solution of the Riemann problems of the GARZ model (see below), it follows that the dynamics of the model never generate values $w\notin [w_\text{min},w_\text{max}]$ or $\rho\notin [0,\rho_\text{max})$. Hence, analogous to the ARZ model (cf.~\cite{AwRascle2000}), the domain $\mathcal{D}_V$ is an invariant region.

\subsubsection{Characteristics and Associated Fields}
\label{subsubsec:characteristics}
The homogeneous part of the GARZ model \eqref{eq:GARZ_conservative} is a conservation law of the form
\begin{equation*}
\vec{U}_t + \vec{F}(\vec{U})_x = 0\;,
\end{equation*}
where
\begin{equation*}
\vec{U} = \begin{pmatrix} \rho \\ q \end{pmatrix}
\quad\text{and}\quad
\vec{F}(\vec{U}) = \begin{pmatrix} u\rho \\ uq \end{pmatrix}\;,
\quad\text{where}\quad
u = V(\rho,q/\rho)\;.
\end{equation*}
The Jacobian of the flux function $\vec{F}(\vec{U})$ is
\begin{equation*}
\nabla\vec{F}(\vec{U})
= \begin{pmatrix} u+\rho\pd{u}{\rho} & \rho\pd{u}{q} \\
q\pd{u}{\rho} & u+q\pd{u}{q} \end{pmatrix}\;,
\end{equation*}
and its eigenvalues and associated eigenvectors are
\begin{align*}
\lambda^{(1)} &= u+\rho\pd{u}{\rho}+q\pd{u}{q} = u+\rho\pd{V}{\rho}
\quad\text{with}\quad
\vec{\gamma}^{(1)} = \begin{pmatrix} \rho \\ q \end{pmatrix}
\quad\text{and thus}\quad
\nabla\lambda^{(1)}\cdot \vec{\gamma}^{(1)} \neq 0\;,
\intertext{and}
\lambda^{(2)} &= u \hspace{10em}
\quad\text{with}\quad
\vec{\gamma}^{(2)}
= \begin{pmatrix} -\pd{u}{q} \\ \phantom{-}\pd{u}{\rho} \end{pmatrix}
\quad\text{and thus}\quad
\nabla\lambda^{(2)}\cdot \vec{\gamma}^{(2)} = 0\;.
\end{align*}
Hence, like the ARZ model \eqref{eq:aw_rascle_zhang_model_conservative}, the GARZ model \eqref{eq:GARZ_conservative} is strictly hyperbolic for $\rho>0$. One of its characteristic velocities, $\lambda^{(1)}$, is slower than the vehicles (i.e., $\lambda^{(1)}<u$) and its associated field is genuinely nonlinear (i.e., it corresponds to shocks and rarefaction waves, see below). Its other characteristic velocity $\lambda^{(2)}$ equals the vehicle velocity and its associated characteristic field is linearly degenerate (i.e., its associated waves are contact discontinuities that are transported with the flow).
\begin{lem}
\label{lem:lambdas_decreasing}
Both characteristic velocities are strictly decreasing w.r.t.\ $\rho$, i.e., $\pd{\lambda^{(1)}}{\rho} < 0$ and $\pd{\lambda^{(2)}}{\rho} < 0$.
\end{lem}
\begin{proof}
We have that $\lambda^{(1)} = V+\rho\pd{V}{\rho} = \pd{Q}{\rho}$. Since $Q$ is assumed concave w.r.t.\ $\rho$, it follows that $\pd{\lambda^{(1)}}{\rho} = \frac{\partial^2 Q}{\partial \rho^2} < 0$. Moreover, $\pd{\lambda^{(2)}}{\rho} = \pd{V}{\rho} < 0$ by Lemma~\ref{lem:concave_Q_decreasing_U}.
\end{proof}
We continue with the discussion of the characteristic fields. The scalar function $I^{(1)} = q/\rho = w$ satisfies $\nabla I^{(1)}\cdot \vec{\gamma}^{(1)} = 0$, and it is a Riemann invariant to $\lambda^{(1)}$. Hence, across waves of the first family, the empty road velocity $w$ is constant. The field associated with the second eigenvalue $\lambda^{(2)}$ is linearly degenerate. It is given by $I^{(1)} = \lambda^{(2)} = u$ and across waves of the second family, the velocity $u$ is constant. The solution to a Riemann problem, i.e., the Cauchy problem to system \eqref{eq:GARZ_conservative} on the  real line with discontinuous piecewise constant initial data $\vec{U}(x,0) = (1-H(x))\vec{U}_\text{L}+H(x)\vec{U}_\text{R}$, where $H$ is the Heaviside function, generalizes naturally from the ARZ model as well. In general, the solution is obtained by superposition of simple waves connecting different constant states: from a given left state $\vec{U}_\text{L}$ to a given right state $\vec{U}_\text{R}$ via an intermediate state $\vec{U}_\text{M}$ that is connected to $\vec{U}_\text{L}$ by a 1-wave (i.e., a Lax-shock or rarefaction associated to the first characteristic field), and to $\vec{U}_\text{R}$ by a 2-wave (i.e., a contact discontinuity). Since the GARZ system \eqref{eq:GARZ_conservative} is---as the ARZ model---of Temple class \cite{Temple1983}, shocks and rarefaction wave curves in phase space coincide.  Moreover, due to Lemma~\ref{lem:lambdas_decreasing}, a simple wave of the first family is either a shock or a rarefaction wave.

In the phase space $(\rho,w)$-plane we discuss the shape of the characteristic fields. The second fields are parallel to the $\rho$-axis, and the first fields are the contours $V(\rho,w) = \text{const}$. Since by assumption $\pd{V}{\rho}<0$ and $\pd{V}{w}>0$, the contours of $V(\rho,w)$ always have a finite and truly positive slope in the $(\rho,w)$-plane. Thus, for any two states $(\rho_\text{L},w_\text{L})$ and $(\rho_\text{R},w_\text{R})$ that satisfy $w_\text{L}\ge u_\text{R}$, where $u_\text{R} = V(\rho_\text{R},w_\text{R})$, there is  a unique intermediate state $(\rho_\text{M},w_\text{M}) = (R(u_\text{R},w_\text{L}), w_\text{L})$, defined via the inverse function given in \S\ref{subsubsec:regions}. Moreover, because $\lambda^{(1)}$ is decreasing with $\rho$ (see Lemma~\ref{lem:lambdas_decreasing}), the Lax entropy conditions \cite{Evans1998} imply that for $\rho_\text{L}<\rho_\text{M}$ the 1-wave is a shock wave (moving with speed $s = \frac{\rho_\text{M}V(\rho_\text{M},w_\text{M}) - \rho_\text{L}V(\rho_\text{L},w_\text{L})}{\rho_\text{M}-\rho_\text{L}}$, given by the Rankine-Hugoniot conditions \cite{Evans1998}), while for $\rho_\text{L}>\rho_\text{M}$ it is a rarefaction wave. The condition $w_\text{L}\ge u_\text{R}$ means that drivers on the left wish to drive at least as fast as the vehicles on the right are driving. If this is not the case, i.e., if $w_\text{L} < u_\text{R}$, then there is no non--negative density at which the 1-wave and the 2-wave intersect. Here, a vacuum state will be generated, analogously to the construction for the ARZ model \cite{AwRascle2000, Rascle2002}. The left state $(\rho_\text{L},w_\text{L})$  is connected by a rarefaction wave to a left vacuum state $(0,w_\text{L})$; this state is connected to a right vacuum state $(0,u_\text{R})$ via another rarefaction (which is feasible because $w_\text{L} < u_\text{R}$); and this then connects to the right state $(\rho_\text{R},w_\text{R})$ via a 2-contact discontinuity.

\subsubsection{Relaxation of GARZ to LWR}
\label{subsubsec:relaxtion_GARZ_LWR}
Smooth solutions to the Cauchy problem of the \emph{inhomogeneous} GARZ model \eqref{eq:GARZ_model_inhomogeneous} relax in time towards solutions of the LWR model \eqref{eq:lighthill_whitham_richards_model}, because $w$ tends to $w_\text{eq}$ along characteristic curves. Note that for general relaxation systems, convergence to a first-order equation is only warranted if a sub-characteristic condition is satisfied, cf.~\cite{Liu1987, ChenLevermoreLiu1994, SeiboldFlynnKasimovRosales2013}. Here, we are in the characteristic case. Moreover, shocks of the $2\times 2$ hyperbolic system \eqref{eq:GARZ_conservative} are also shocks of the LWR model \eqref{eq:lighthill_whitham_richards_model}. Therefore, for the Cauchy problem, GARZ solutions converge to LWR solutions as $t\to\infty$, and/or as $\tau\to 0$.

\begin{figure}
\centering
\begin{minipage}[b]{.90\textwidth}
\centering
\includegraphics[width=.75\textwidth]{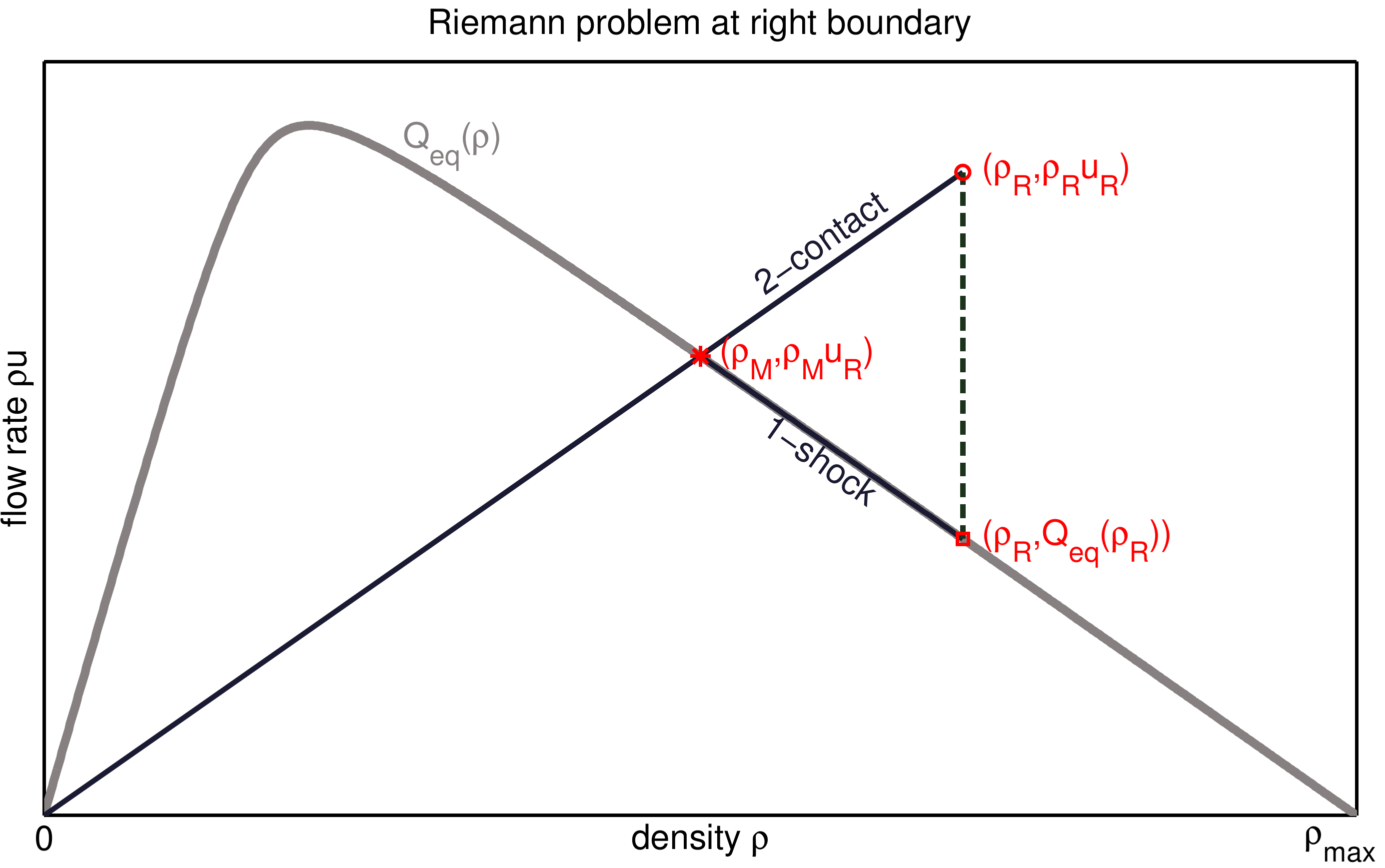}
\vspace{-.4em}
\caption{Riemann problem at right domain boundary. Consider a constant initial state $(\rho_\text{R},Q_\text{eq}(\rho_\text{R}))$ in the domain, and prescribed boundary data $(\rho_\text{R},\rho_\text{R}u_\text{R})$ with $u_\text{R} > U_\text{eq}(\rho_\text{R})$. The LWR model preserves the constant state (the boundary data is projected onto the equilibrium curve vertically along the dashed line). In contrast, the GARZ model generates a new state $(\rho_\text{M},Q_\text{eq}(\rho_\text{M}))$ at the boundary (projected onto the equilibrium curve along a ray through the origin), from which a shock moves into the domain, thus changing the initial state to the new boundary state.}
\label{fig:Riemann_problem_bc}
\end{minipage}
\end{figure}

In contrast, for initial-boundary-value problems on a bounded domain $x\in [x_\text{L},x_\text{R}]$, this last property is in general not true. To highlight this fact we consider a simplified setting depicted in Fig.~\ref{fig:Riemann_problem_bc}. Let constant boundary data $\rho(x_\text{L})$, $u(x_\text{L})$, $\rho(x_\text{R})$, $u(x_\text{R})$ be given. With this data, we can solve the LWR model \eqref{eq:lighthill_whitham_richards_model} and the inhomogeneous GARZ model \eqref{eq:GARZ_conservative}. Then, in general solutions to the latter problem do not converge to solutions of the former, even in the limit $\tau\to 0$. Consider a constant state $(\rho_\text{R},Q_\text{eq}(\rho_\text{R}))$ inside the domain. At the outflow boundary $x_\text{R}$, let a state $(\rho_\text{R},\rho_\text{R}u_\text{R})$ be given where $u_\text{R} > U_\text{eq}(\rho_\text{R})$. The LWR model only uses the density information, and thus the constant state is preserved. In contrast, the GARZ model uses the full state in the Riemann problem. Its solution yields an intermediate state $(\rho_\text{M},Q_\text{eq}(\rho_\text{M}))$, whose density is determined via the relation $\rho_\text{M}u_\text{R} = Q_\text{eq}(\rho_\text{M})$. This intermediate state connects to the boundary state via a contact discontinuity (a 2-wave) moving with speed $u_\text{R}$, and to the interior state via a shock (a 1-wave) that moves with speed
\begin{equation*}
s = \frac{Q_\text{eq}(\rho_\text{R})-Q_\text{eq}(\rho_\text{M})}
{\rho_\text{R}-\rho_\text{M}}\;,
\end{equation*}
which in the situation depicted in Fig.~\ref{fig:Riemann_problem_bc} moves \emph{into} the domain, since $s<0$. Thus, after some time in the GARZ model the intermediate state $(\rho_\text{M},\rho_\text{M}u_\text{R})$ is observed within the domain. Note that this argument holds independent of the value of the relaxation time $\tau$ in the model \eqref{eq:GARZ_model_inhomogeneous}.

\vspace{1.5em}
\section{Data-Fitted Traffic Models}
\label{sec:data-fitted_models}
In this section we describe how the parameter functions of the traffic models presented in \S\ref{sec:models} can be fitted to historic fundamental diagram data. We assume that flow rate vs.\ density pairs $(\rho_j,Q_j),\;j=1,\dots n$ are given from long-term measurements (commonly obtained via stationary sensors). As visible in the right panel of Fig.~\ref{fig:fd_lwr_arz}, these data (gray dots) tend to exhibit a relatively clear functional relationship between $\rho$ and $Q$ for low densities. In turn, for medium densities, a significant spread is visible, i.e., a single $\rho$-value corresponds to many different flow rates $Q$. Finally, for large densities, very few data points are available at all.

\subsection{Data-Fitting for the LWR and ARZ Models}
The first-order LWR model \eqref{eq:lighthill_whitham_richards_model} must represent these data via a single function $Q(\rho)$. As the spread of the data cannot be captured, it is reasonable to find a function that lies ``in the middle'' of the cloud of data points. Specifically, we employ the approach presented in \cite{FanSeibold2013}. First, since the stagnation density $\rho_\text{max}$ is not represented well via data, we prescribe it as a fixed constant, given by a typical vehicle length of 5 meters, plus 50\% of additional safety distance,
\begin{equation*}
\rho_\text{max} = \frac{\text{number of lanes}}
{\text{typical vehicle length}\times\text{safety distance factor}}
= \frac{\#\text{lanes}}{7.5\text{m}}\;.
\end{equation*}

Second, a three-parameter family of smooth and strictly concave flow rate curves is selected as
\begin{equation}
\label{eq:flow_rate_curve}
Q_{\alpha,\lambda,p}(\rho) = \alpha\prn{a+(b-a)\rho/\rho_\text{max}-\sqrt{1+y^2}}\;,
\end{equation}
where
\begin{align*}
a = \sqrt{1+\left(\lambda p\right)^2}\;, \quad
b = \sqrt{1+\left(\lambda(1-p)\right)^2}\;,\text{~and}\quad
y = \lambda \left(\rho/\rho_\text{max}-p\right).
\end{align*}
Each flow rate function $Q_{\alpha,\lambda,p}(\rho)$ in this family vanishes for $\rho = 0$ and $\rho = \rho_\text{max}$. The three free parameters allow for controlling three important features of $Q_{\alpha,\lambda,p}(\rho)$: the value of maximum flow rate $Q_\text{max}$ (mainly determined by $\alpha$), the critical density $\rho_\text{c}$ (mainly controlled by $p$), and the ``roundness" of the curve, i.e., how rapidly the slope transitions from positive to negative near $\rho_\text{c}$ (dominated by $\lambda$).

Third, from this three-parameter family of flow rate curves, the one is selected that is the closest to the data points $(\rho_j,Q_j),\;j=1,\dots n$ in a least-squares sense, i.e. we solve
\begin{equation}
\label{eq:LSQ}
\min_{\alpha,\lambda,p}\; \sum_{j=1}^n
(Q_{\alpha,\lambda,p}(\rho_{j})-Q_{j})^2\;.
\end{equation}
In the right panel of Fig.~\ref{fig:fd_lwr_arz}, the resulting least-squares fit to the given gray data points, called $Q_\text{eq}(\rho)$, is depicted by the red curve. The red curve in the left panel represents the resulting velocity function $U_\text{eq}(\rho) = Q_\text{eq}(\rho)/\rho$.

As described in \S\ref{subsec:ARZ_generalization_of_LWR}, the ARZ model \eqref{eq:aw_rascle_zhang_model_w} generalizes the LWR model to a one-parameter family of velocity curves $u_w(\rho) = U_\text{eq}(\rho)+(w-U_\text{eq}(0))$ (black curves in the left panel of Fig.~\ref{fig:fd_lwr_arz}) and flow rate curves $Q_w(\rho) = Q_\text{eq}(\rho)+(w-U_\text{eq}(0))\rho$ (right panel of Fig.~\ref{fig:fd_lwr_arz}). Due to this property, the ARZ model possesses the same amount data-fitted parameters as the LWR model.

An interpretation of the ARZ family of velocity curves is that different $w$-values represent different types of drivers; the larger $w$, the faster the corresponding drivers tend to drive. As motivated in \S\ref{subsec:garz}, this captures the spread in the fundamental diagram (which is desirable), but it also results in vastly varying stagnation densities for different types of drivers. This last property is unrealistic, as the maximum density is a property of the road, rather than of the behavior of drivers \cite{FanSeibold2013}. We therefore need to construct a family of curves that are not simple shifts of each other, as done below.

\begin{figure}
\begin{minipage}[b]{.480\textwidth}
\includegraphics[width=\textwidth]{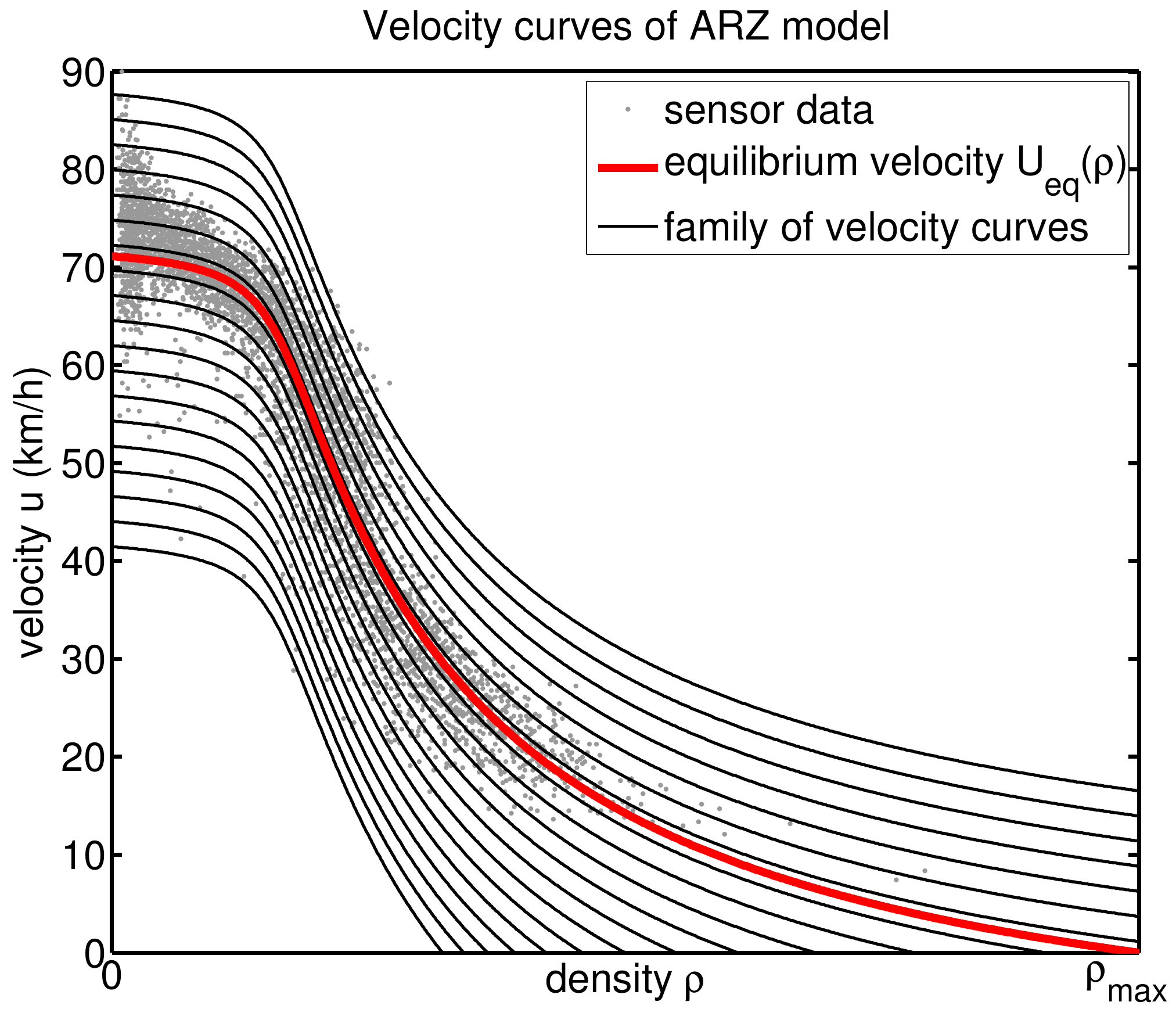}
\end{minipage}
\hfill
\begin{minipage}[b]{.494\textwidth}
\includegraphics[width=\textwidth]{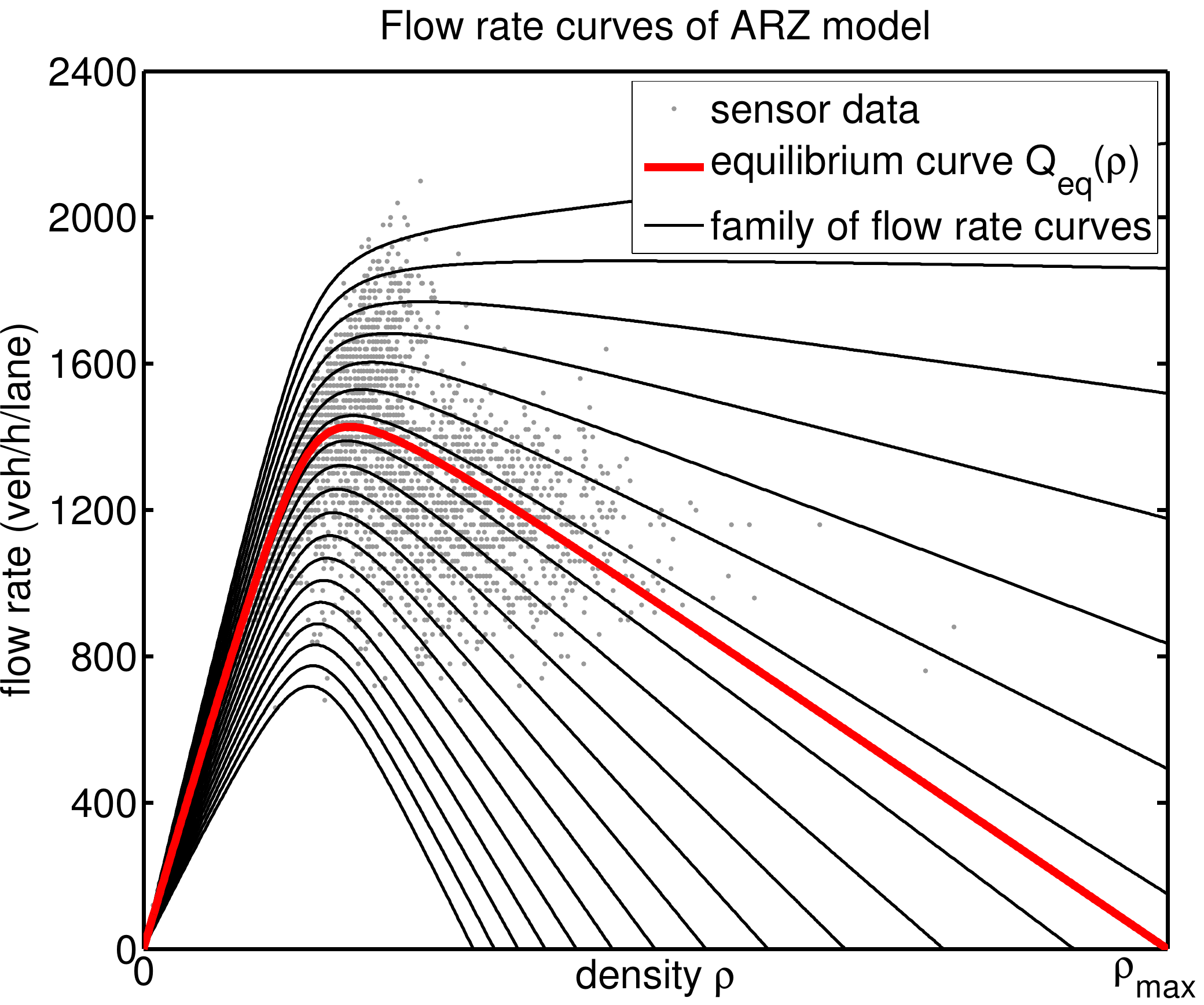}
\end{minipage}
\vspace{-.4em}
\caption{Velocity vs.\ density (left panel) and flow rate vs.\ density (right panel) curves of the smooth three-parameter model \eqref{eq:flow_rate_curve}, fitted with historic fundamental diagram data (gray dots), for the ARZ model.}
\label{fig:fd_lwr_arz}
\end{figure}

\subsection{Data-Fitting for the GARZ Model}
\label{sec:data-fitting_GARZ}
The GARZ model \eqref{eq:GARZ_model} is based on a generalized velocity function $V(\rho,w)$, where---as for the ARZ model---$w$ represents different types of drivers, and thus parameterizes a families of velocity and flow rate curves, respectively. We construct these families of curves by generalizing the least-squares fit \eqref{eq:LSQ} to a weighted least-squares fit, as follows.

Given a weight parameter $0<\beta<1$, we consider the minimization problem
\begin{equation}
\label{eq:WLSQ}
\min_{\alpha,\lambda,p}\left\{
(1-\beta) \sum^{n}_{j=1}{\left((Q_{\alpha,\lambda,p}(\rho_{j})-Q_{j})_{+}\right)^2}+
\beta \sum^{n}_{j=1}{\left((Q_{\alpha,\lambda,p}(\rho_{j})-Q_{j})_{-}\right)^2}
\right\}\;,
\end{equation}
where
\begin{align*}
\left(Q_{\alpha,\lambda,p}(\rho_{j})-Q_{j}\right)_{+}\
&= \max\left\{ Q_{\alpha,\lambda,p}(\rho_{j})-Q_{j},0\right\}\;, \\
\left(Q_{\alpha,\lambda,p}(\rho_{j})-Q_{j}\right)_{-}\
&= \max\left\{-Q_{\alpha,\lambda,p}(\rho_{j})+Q_{j},0\right\}\;.
\end{align*}
For $\beta=\frac{1}{2}$, problem \eqref{eq:WLSQ} reduces to \eqref{eq:LSQ}, i.e., the LWR equilibrium curve is recovered. For $\beta<\frac{1}{2}$, data below the curve is penalized more, and consequently the resulting curve moves downwards. In turn, if $\beta>\frac{1}{2}$, curves above the equilibrium curve are obtained.

The weighted least-squares problem \eqref{eq:WLSQ} generates a one-parameter family of curves $Q_\beta(\rho) = Q_{\alpha(\beta),\lambda(\beta),p(\beta)}(\rho)$, parameterized by $\beta$; and consequently it also generates a family of velocity functions
\begin{equation*}
V_\beta(\rho) = \begin{cases}
Q_\beta(\rho)/\rho &\text{if~}\rho>0 \\
\frac{\partial Q_\beta}{\partial\rho}(0) &\text{if~}\rho=0 \end{cases}\;.
\end{equation*}
In this paper, we restrict to the case that the velocity curves in the family are non-intersecting, i.e.,
\begin{equation}
\label{eq:non_intersecting}
\text{If~}\beta_1<\beta_2\;,\text{~~then~~}
V_{\beta_1}(\rho)<V_{\beta_2}(\rho)
\text{~~for~~}\rho\in [0,\rho_{\text{max}})\;.
\end{equation}
Note that in general (i.e., for a general family of flow rate functions, and for general data points), property \eqref{eq:non_intersecting} is not necessarily satisfied; and furthermore, problem \eqref{eq:WLSQ} need not have a unique solution. However, in all cases studied in this paper, the minimization problem \eqref{eq:WLSQ} does have a unique solution; and property \eqref{eq:non_intersecting} is in fact satisfied.

\begin{figure}
\begin{minipage}[b]{.480\textwidth}
\includegraphics[width=\textwidth]{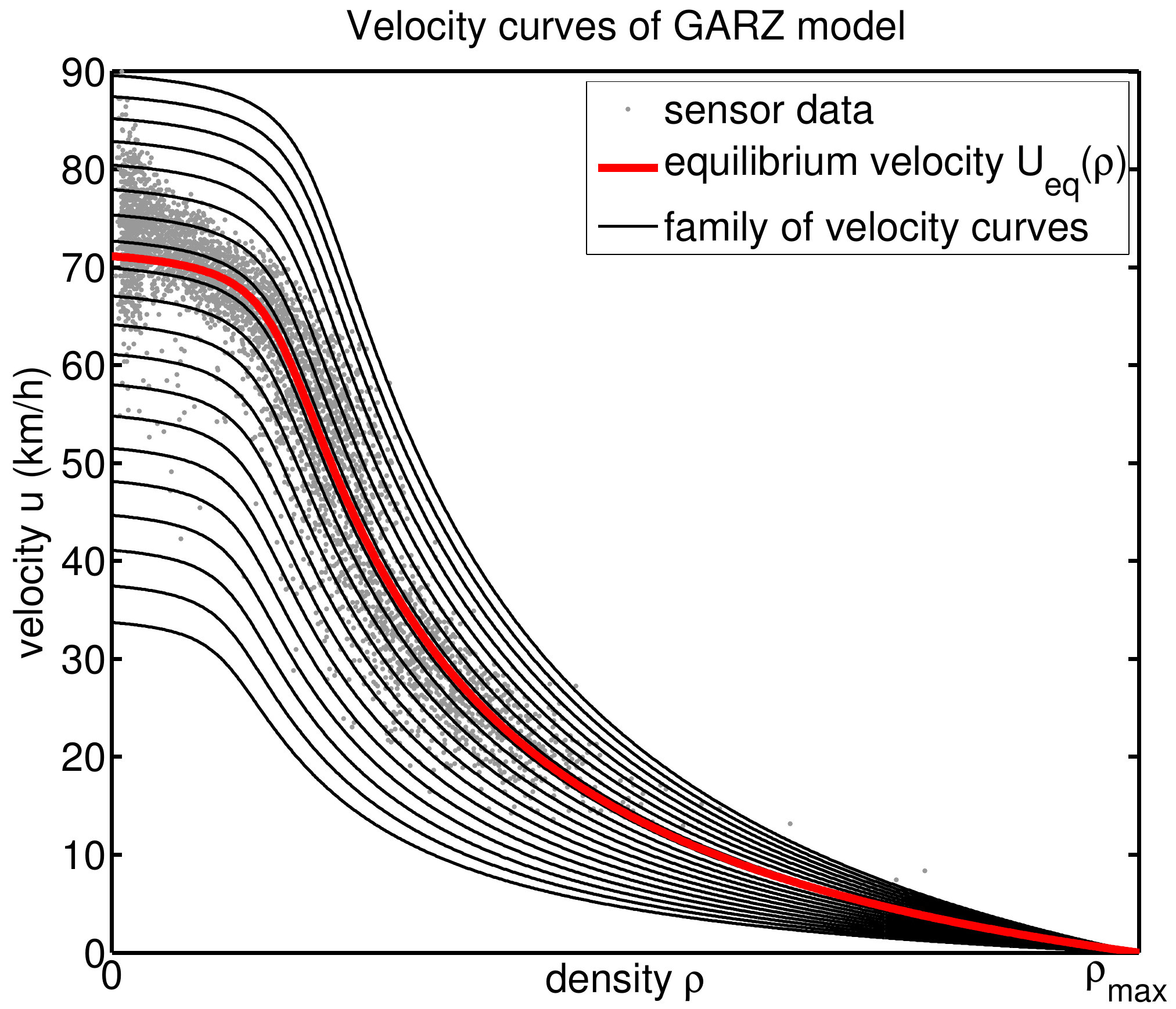}
\end{minipage}
\hfill
\begin{minipage}[b]{.494\textwidth}
\includegraphics[width=\textwidth]{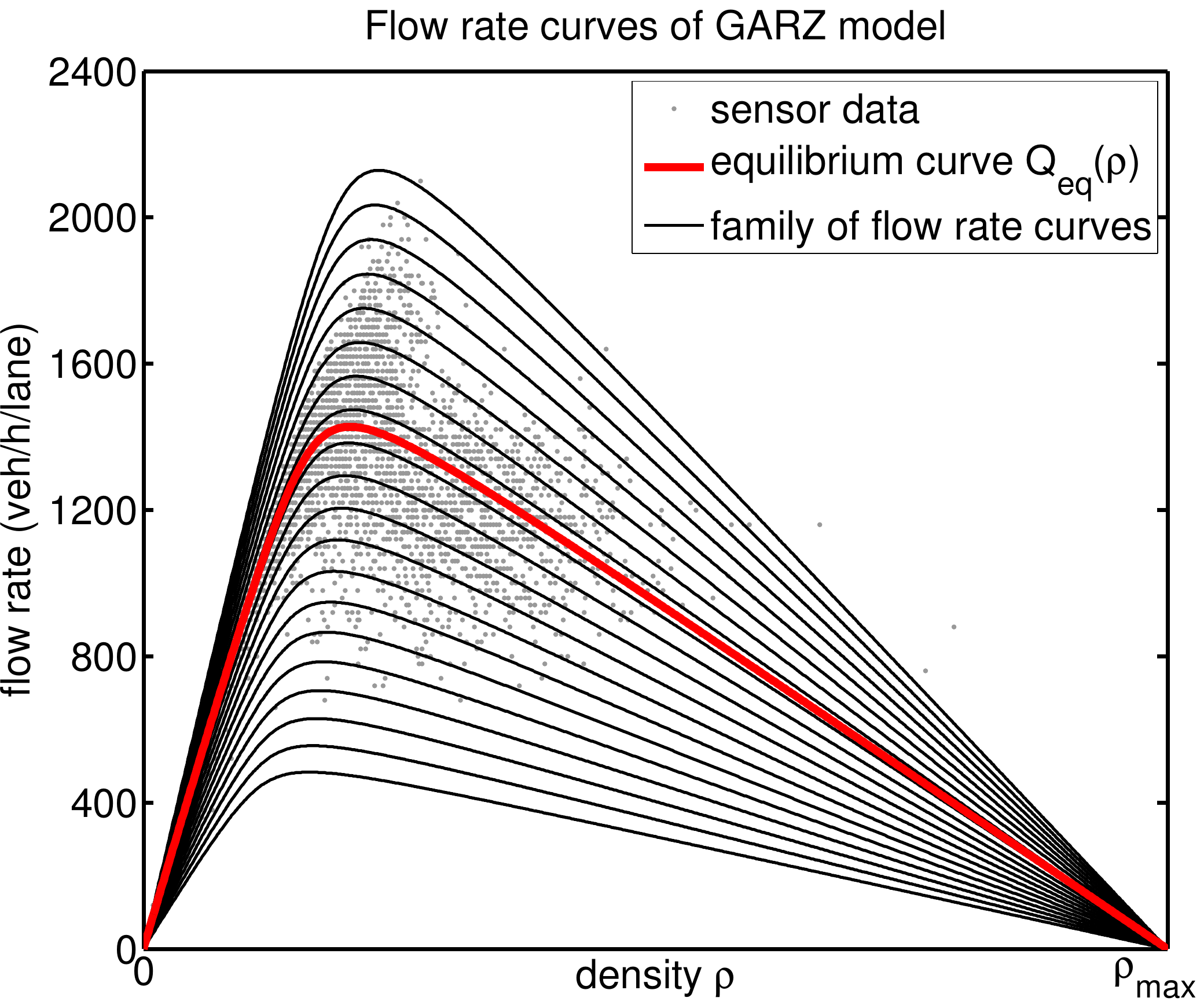}
\end{minipage}
\vspace{-.4em}
\caption{Family of velocity vs.\ density (left panel) and flow rate vs.\ density (right panel) curves generated from the weighted least square (WLSQ) algorithm in constructing the velocity function $u = V(\rho,w)$ in the GARZ model.}
\label{fig:curves_garz}
\end{figure}

While problem \eqref{eq:WLSQ} is defined for all $0<\beta<1$, values of $\beta$ that are extremely close to 0 or 1 tend to lead to unreasonable curves. The reason is that in the limit $\beta\nearrow 1$, the resulting curve is the lowest curve that has no data points above it, and as a result it adjusts to outliers in the data (similar arguments hold for $\beta\searrow 0$). We therefore define a lower/upper flow rate curve, such that 99.9\% of all data points lie above/below it. Consequently, together with the equilibrium curve, we have the following three special flow rate curves
\begin{equation*}
Q_\text{min}(\rho) = Q_{\beta_\text{min}}(\rho)\;, \quad
Q_\text{eq }(\rho) = Q_{\frac{1}{2}}(\rho)\;,\;\text{and} \quad
Q_\text{max}(\rho) = Q_{\beta_\text{max}}(\rho)\;,
\end{equation*}
where here we use $\beta_\text{min} = 10^{-4}$ and $\beta_\text{max} = 1-10^{-4}$. In the right panel of Fig.~\ref{fig:curves_garz} these three curves are depicted by the lowest black curve, the red curve, and the uppermost black curve, respectively.

Even though the parameter $\beta$ defines a family of flow rate curves as desired, it has the shortcoming that it does not have an immediate interpretation as a property of traffic flow. We therefore re-parameterize the family in terms of the empty road velocity $w$, as follows. For any $\beta\in [\beta_\text{min},\beta_\text{max}]$, we define $w$ as the resulting slope of the curve $Q_\beta$ at $\rho = 0$, i.e.,
\begin{equation*}
w = w(\beta) = V_\beta(0)\;.
\end{equation*}
Due to property \eqref{eq:non_intersecting}, the relationship $w = w(\beta)$ is strictly increasing, and thus can be inverted into $\beta = \beta(w)$, defined on the interval $w\in [w_\text{min},w_\text{max}]$, where $w_\text{min} = Q_\text{min}'(0)$ and $w_\text{max} = Q_\text{max}'(0)$. Using this re-parameterization, we obtain a generalized flow rate function
\begin{equation*}
Q(\rho,w) = Q_{\beta(w)}(\rho)\;,
\end{equation*}
and a generalized velocity function
\begin{equation*}
V(\rho,w) = V_{\beta(w)}(\rho)\;,
\end{equation*}
as used in the GARZ model \eqref{eq:GARZ_model}. The properties of $V(\rho,w)$ assumed in \S\ref{subsec:garz} are satisfied by construction.

\subsection{Domain Extension for the GARZ Model}
The systematical construction of a generalized velocity function $V(\rho,w)$, presented in \S\ref{sec:data-fitting_GARZ}, is in line with the regions of the GARZ variables defined in \S\ref{subsubsec:regions}. In particular, the function $W(\rho,u)$ is defined only for $V(\rho,w_\text{min})\le u\le V(\rho,w_\text{max})$. However, when applying the GARZ model in a forward computation, velocity data may be provided through initial and boundary conditions that lie outside of the domain of $W(\rho,u)$. In order to make sense of the model for such data, we effectively extend the domain of the function $W(\rho,u)$ via a projection of such data, as follows.

Given a density--velocity pair $(\rho,u)$, where $0<\rho<\rho_\text{max}$, we define a projected velocity as
\begin{equation*}
\tilde{U}(\rho,u) = \min\{\max\{u,V(\rho,w_\text{min})\},V(\rho,w_\text{max})\}\;,
\end{equation*}
and thus obtain the extended function
\begin{equation*}
\tilde{W}(\rho,u) = W(\rho,\tilde{U}(\rho,u))
\end{equation*}
that is defined for arbitrary velocity values. This simple projection (densities are left unchanged, and velocities are moved onto the lowest or highest curve, respectively) provides a constant extension of the function $W$ beyond its domain $\mathcal{D}_W$. Note that the range of $W$ remains unaffected as being $[w_\text{min},w_\text{max}]$, and consequently the function $\tilde{W}$ is not invertible outside of $\mathcal{D}_W$.

\vspace{1.5em}
\section{Numerical Methods}
\label{sec:numerical_methods}
All models are approximated numerically using a finite volume method on a regular grid of cell size $\Delta x$ and time step $\Delta t$, chosen so that the CFL condition \cite{CourantFriedrichsLewy1928}
\begin{equation*}
s_\text{max} \Delta t \leq \Delta x\;,
\end{equation*}
is satisfied, where $s_\text{max} = \max_k |\lambda_k|$ is the largest wave speed (see \S\ref{subsubsec:characteristics} for the characteristic velocities of the models). In all examples throughout this paper, the grid resolution is chosen small enough ($\Delta x\le 50\text{cm}$) so that the numerical approximation errors are much smaller than the model errors. Hence, the studies are conducted truly on the continuum level.

The first order model \eqref{eq:lighthill_whitham_richards_model} is solved using Godunov's method \cite{Godunov1959}. For the second order models, we have to account for the fact that the inhomogeneous GARZ model \eqref{eq:GARZ_conservative} becomes stiff if $\tau$ is small. Hence, we employ a semi-implicit finite volume scheme that treats the nonlinear hyperbolic terms explicitly and the relaxation terms implicitly (to prevent a time step restriction $\Delta t = O(\tau)$). The update rule of a state $(\rho_j^n,q_j^n)$ in cell $j$ from time $t_n$ to the state $(\rho_j^{n+1},q_j^{n+1})$ at time $t_{n+1} = t_n+\Delta t$ reads as
\begin{align}
\label{eq:update_rho}
\rho_j^{n+1} &= \rho_j^n - \frac{\Delta t}{\Delta x}\!
\prn{(\mathcal{F}_\rho)_{j+\frac{1}{2}}^n-(\mathcal{F}_\rho)_{j-\frac{1}{2}}^n}\;, \\
\label{eq:update_q}
q_j^{n+1} &= q_j^n - \frac{\Delta t}{\Delta x}\!
\prn{(\mathcal{F}_q)_{j+\frac{1}{2}}^n-(\mathcal{F}_q)_{j-\frac{1}{2}}^n}
+\frac{\Delta t}{\tau}\!
\prn{Q_\text{eq}(\rho_j^{n+1})-Q(\rho_j^{n+1},q_j^{n+1}/\rho_j^{n+1})}\;.
\end{align}
Here $(\mathcal{F}_\rho)_{j+\frac{1}{2}}^n = \mathcal{F}_\rho(\rho_j^n,q_j^n,\rho_{j+1}^n,q_{j+1}^n)$ denotes the numerical flux for the quantity $\rho$ through the boundary between cells $j$ and $j+1$; the other fluxes are defined accordingly. Moreover, $Q(\rho,w) = \rho V(\rho,w)$ is the model's two-parameter flow rate function, and $Q_\text{eq}(\rho) = \rho U_\text{eq}(\rho) = Q(\rho,w_\text{eq})$ is the equilibrium flow rate function.

As in the Godunov method \cite{Godunov1959} one could use the exact solution to the Riemann problem (cf.~\cite{AwRascle2000, Fan2013}) to define the numerical fluxes. However, this would require the inversion of the velocity function $u = V(\rho,w)$, which is costly for the GARZ model. A less expensive approach, employed here, is to define the numerical fluxes via the HLL approximate Riemann solver \cite{HartenLaxVanLeer1983}, which approximates the true Riemann problem by a single constant intermediate region. Note that due to the fine grid resolution, and due to the fact that initial and boundary conditions are continuous, spurious overshoots that may occur in the velocity (cf.~\cite{ChalonsGoatin2007}) are negligibly small.

Since the inhomogeneous model \eqref{eq:GARZ_conservative} possesses a relaxation only in the momentum equation, the time update of the density variable, given by \eqref{eq:update_rho}, is fully explicit. Therefore, in the time update of the generalized momentum, given by \eqref{eq:update_q}, the quantity $\rho_j^{n+1}$ is known. Specifically, the numerical scheme is implemented in three steps:
\begin{enumerate}[ 1)]
\item
Based on the data $(\rho_j^n,q_j^n)\,\forall j$ at time $t_n$, the fluxes $((\mathcal{F}_\rho)_{j+\frac{1}{2}}^n, (\mathcal{F}_q)_{j+\frac{1}{2}}^n)\,\forall j$ are computed.
\item
The new density states $\rho_j^{n+1}\,\forall j$ are computed via \eqref{eq:update_rho}.
\item
The new generalized momenta $q_j^{n+1}\,\forall j$ are computed according to \eqref{eq:update_q}. On the cell $j$, the new state $q_j^{n+1}$ is obtained as the solution of the scalar nonlinear equation $G(q) = 0$, where
\begin{equation}
\label{eq:nonlinear_equation}
G(q) = q + \frac{\Delta t}{\tau} Q(\rho_j^{n+1},q/\rho_j^{n+1})
- q_j^n + \frac{\Delta t}{\Delta x}\!
\prn{(\mathcal{F}_q)_{j+\frac{1}{2}}^n-(\mathcal{F}_q)_{j-\frac{1}{2}}^n}
-\frac{\Delta t}{\tau}Q_\text{eq}(\rho_j^{n+1})\;.
\end{equation}
The root of \eqref{eq:nonlinear_equation} is found up to machine accuracy within a few Newton steps, using $q_j^n$ as the starting guess.
\end{enumerate}
In the special case of the ARZ model \eqref{eq:aw_rascle_zhang_model_conservative}, the update \eqref{eq:update_q} is given by the explicit formula
\begin{equation}
q_j^{n+1}
= \frac{ q_j^n - \frac{\Delta t}{\Delta x}\!
\prn{(\mathcal{F}_q)_{j+\frac{1}{2}}^n-(\mathcal{F}_q)_{j-\frac{1}{2}}^n}
+\frac{\Delta t}{\tau} \rho_j^{n+1} w_\text{eq} }{1+\frac{\Delta t}{\tau}}\;.
\end{equation}
It should further be remarked that the semi-implicit scheme \eqref{eq:update_rho} and \eqref{eq:update_q} is equivalent to the fractional step approach that first approximates the homogeneous part of \eqref{eq:GARZ_conservative} via a forward Euler step, and then approximates the relaxation part via a backward Euler step. Hence, in the limit $\tau\to 0$ (while $\Delta t$ fixed), the scheme amounts to simply projecting $q$ onto the equilibrium curve in the relaxation step, i.e., equation \eqref{eq:update_q} turns into $q_j^{n+1} = Q_\text{eq}(\rho_j^{n+1})$. In turn, in the homogeneous case, i.e., $\tau\to\infty$, the relaxation terms are simply omitted.

The boundary data are provided by introducing a ghost cell adjacent to the outermost grid cell (on either side of the domain), in which the boundary state $(\rho,q)$ is assumed. The numerical fluxes in \eqref{eq:update_rho} and \eqref{eq:update_q} then by construction pick up the information corresponding to waves that enter the computational domain.

\vspace{1.5em}
\section{Validation and Comparison of Models via Real Data}
\label{sec:validation}
In the following, we validate the presented models by studying how well they reproduce the evolution of real traffic data, and in that process we compare the predictive accuracy of the models. A particular focus lies on the investigation of the extent to which the GARZ model, that addresses various shortcoming of traditional models, improves the actual model agreement with real traffic data. We conduct the validations using the NGSIM trajectory data set \cite{TrafficNGSIM} and the RTMC sensor data set \cite{TrafficMnDOT}.

The test framework considered here is based on the methodology presented in \cite{FanSeibold2013} and further developed in \cite{FanSeibold2014}: on a segment of highway, a three-detector test problem \cite{Daganzo1997} is formulated. At each end of the segment, the traffic state is (at all times) provided to the traffic model, which is advanced forward in time (using the numerical methods described in \S\ref{sec:numerical_methods}) inside the segment. The predictions that the traffic model produces in time are then compared to real data inside the segment, and the deviation between predicted and real traffic states is used to quantify the model error.

\subsection{Treatment of Data}
\label{subsec:data_treatment}
As described in \cite{FanSeibold2013}, continuous field quantities $\rho(x,t)$ and $u(x,t)$ are constructed from the NGSIM vehicle trajectory data \cite{TrafficNGSIM_I80}, using kernel density estimation \cite{Parzen1962, Rosenblatt1956}. In this approach, given vehicle locations $x_j(t)$ (including ``ghost vehicle'' positions, obtained via reflection at the boundaries, see \cite{KarunamuniAlberts2005}), density and flow rate functions are obtained as superpositions of Gaussian profiles,
\begin{equation*}
\rho(x,t) = \sum_{j=1}^n K(x-x_j(t))
\quad\text{and}\quad
Q(x,t) = \sum_{j=1}^n u_j K(x-x_j(t))\;,
\text{~where~}
K(x) = \tfrac{1}{\sqrt{2\pi}h}e^{-\frac{x^2}{2h^2}}\;,
\end{equation*}
and the velocity field is then given by $u(x,t) = Q(x,t)/\rho(x,t)$. The kernel width is chosen $h = 25$\;meters. These field quantities then define initial conditions ($t=0$) and boundary conditions (when evaluated at the segment boundary positions) for the traffic models, and reference states for the validation (inside the segment for $t>0$). Before the boundary data can be provided to the traffic model, one additional processing step must be applied to address spurious fast oscillations in the reconstructed boundary data (due to variations in the starting and end position of each vehicle trajectory in the data set): the time domain is divided into intervals of length 15\;seconds, and on each interval the boundary data is replaced by a cubic polynomial that is a least-squares fit to the data, under the constraint that the resulting evolution is globally $C^1$.

For the RTMC sensor data, vehicles densities and flow rates are given at three sensor positions, aggregated in intervals of length 30\;seconds. Temporally continuous quantities $\rho(x_s,t)$ and $u(x_s,t)$ at a sensor position $x_s$ are generated via cubic spline interpolation (in time) of the aggregated information. One shortcoming of the RTMC data is the absence of a reliable initial state (because information is given only at the sensor positions). This problem is circumvented by running the models forward through an initialization phase (5~minutes), before the actual model comparison is started. During this phase, the boundary data has time to move into the domain and create a reasonable initial state for the actual validation.

\subsection{Quantification of Model Errors}
\label{sec:error}
The quantification of the deviation of the model predictions from the actual data requires two aspects to be specified: first, which field quantities to consider and how to combine them into a single quantity; and second, if data is available at multiple positions and/or times, how to combine these multiple pieces of information into a single quantity?

Regarding the choice of field quantities, in this study we are interested in models that predict traffic densities (as required for instance for ramp metering) and velocities (as required for instance for travel time estimates) accurately. Since densities and velocities have different physical units, suitable normalization constants must be found, so that the deviations in each quantity contribute with equal influence to the total model error.

Given model predictions $\rho^\text{model}(x,t)$ and $u^\text{model}(x,t)$, and real data $\rho^\text{data}(x,t)$ and $u^\text{data}(x,t)$, we define a space-and-time-dependent error measure as
\begin{equation}
\label{eq:error_measure}
E(x,t) = \frac{|\rho^\text{model}(x,t)-\rho^\text{data}(x,t)|}{\Delta\rho}
+\frac{|u^\text{model}(x,t)-u^\text{data}(x,t)|}{\Delta u}\;,
\end{equation}
where the normalization constants $\Delta\rho$ and $\Delta u$ represent the ranges of the fundamental diagram data, as defined below. Note that various choices of normalization constants have been proposed in the literature. For instance, in \cite{BlandinWorkGoatinPiccoliBayen2011} the absolute errors in density and velocity are scaled with $\|\rho^\text{data}(x,t)\|_{L^1}$ and $\|u^\text{data}(x,t)\|_{L^1}$, respectively. A shortcoming of this choice is that for traffic flow at low densities, errors in density get divided by a very small number and thus significantly amplified. An alternative choice is employed in \cite{FanSeibold2013} by using $\rho_\text{max}$ and $u_\text{max}$ as normalization constants. However, these tend to give too much influence to velocity errors, because even in moving congested traffic flow, $\rho/\rho_\text{max}$ tends to be significantly smaller than $u/u_\text{max}$.

\begin{figure}
\centering
\begin{minipage}[b]{.90\textwidth}
\centering
\includegraphics[width=.90\textwidth]{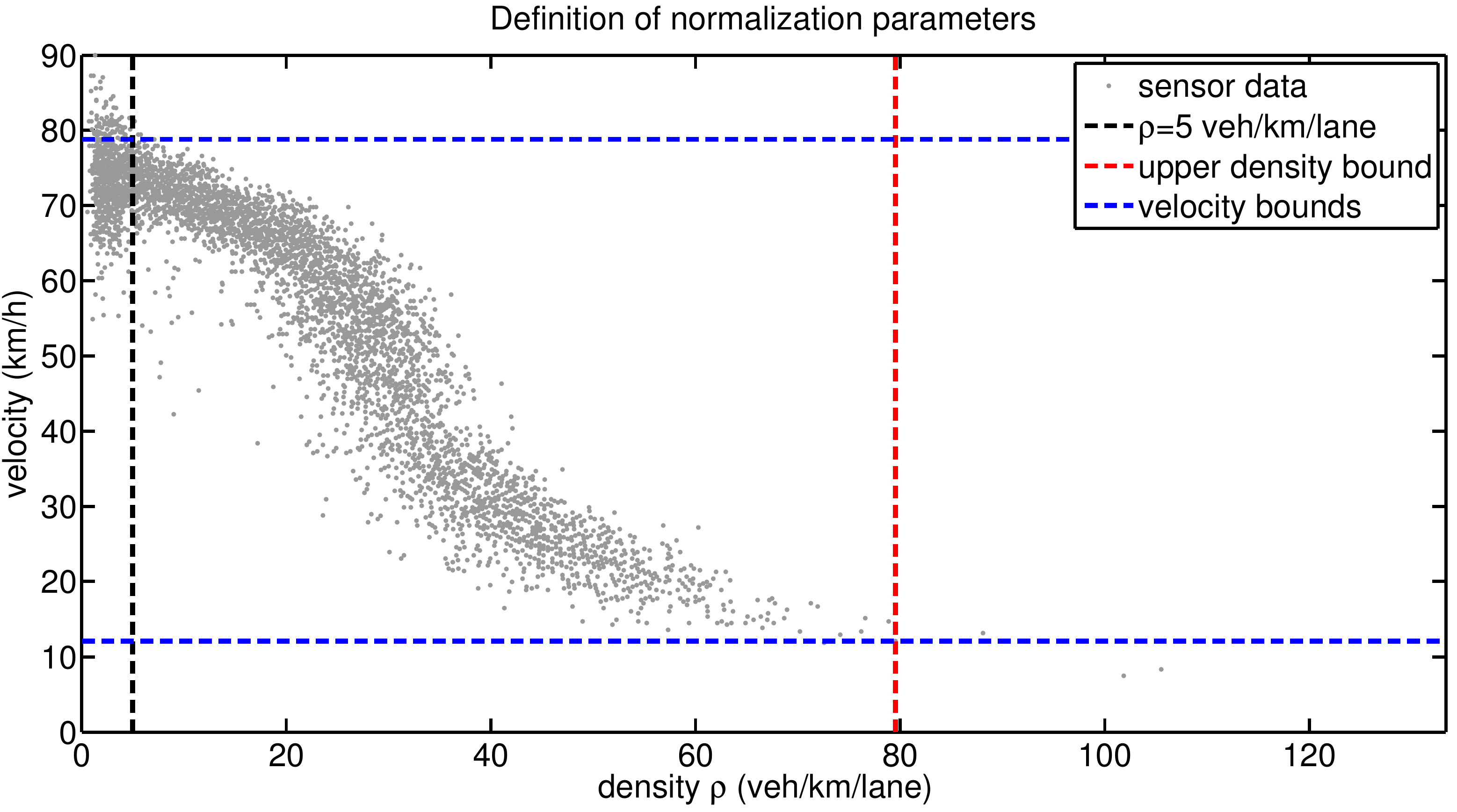}
\vspace{-.4em}
\caption{Construction of ranges in density and velocity in historic fundamental diagram data, using an upper density (red line) and lower and upper velocity boundaries (blue lines). Data points with densities below a threshold (black line), as well as outliers, are systematically excluded.}
\label{fig:normalization}
\end{minipage}
\end{figure}

In line with \cite{FanSeibold2014}, we argue that balanced weights are given when the error in each field quantity is related to the maximum variation that the respective quantity exhibits in the historic fundamental diagram. In order to exclude the influence of outliers in the data, we conduct the following four-step approach. First, all data points $(\rho_j,u_j)$ with $\rho_j<5\,\text{veh}/\text{km}/\text{lane}$ are neglected. The rationale is that these data do not contribute any meaningful information about the spread in the traffic states, and moreover such low density values are not meaningful in the context of a macroscopic description of traffic flow. In Fig.~\ref{fig:normalization}, this boundary is depicted by the black line. Second, similar to the method presented in \cite{BlandinBrettiCutoloPiccoli2009}, the upper density boundary $\rho^\text{up}$ is defined such that 99.9\% of the remaining data points lie below it (red line in Fig.~\ref{fig:normalization}). Third, the lower (upper) velocity boundary $u^\text{low}$ ($u^\text{up}$) is defined such that 99.9\% of the remaining data points lie above (below) it (blue lines in Fig.~\ref{fig:normalization}). Fourth, we define the data ranges
\begin{equation*}
\Delta\rho = \rho^\text{up}
\quad\text{and}\quad
\Delta u = u^\text{up}-u^\text{low}\;.
\end{equation*}

Regarding the norms and averages to measure the model errors, we use the following expressions. On a segment $x\in [0,L]$ and time interval $t\in [0,T]$, spatial and spatio-temporal averages are considered
\begin{align}
E^{x}(t) &= \frac{1}{L}\int_0^L E(x,t)\ud{x}\;, \label{eq:error_x} \\
E &= \frac{1}{TL}\int_0^T\int_0^L E(x,t)\ud{x}\ud{t}\;. \label{eq:error_xt}
\end{align}
Moreover, for the RTMC data set, the temporal error at a sensor position $x_s$ inside the road segment on a given day is considered, as well averages over multiple days
\begin{align}
E_\text{day} &= \frac{1}{T}\int_0^T E(x_s,t)\ud{t}\;, \label{eq:error_t} \\
E &= \frac{1}{\#\text{days}}\sum_{\text{day}=1}^{\#\text{days}} E_\text{day}\;. \label{eq:error_day}
\end{align}

\begin{figure}
\begin{minipage}[b]{.32\textwidth}
\includegraphics[width=\textwidth]{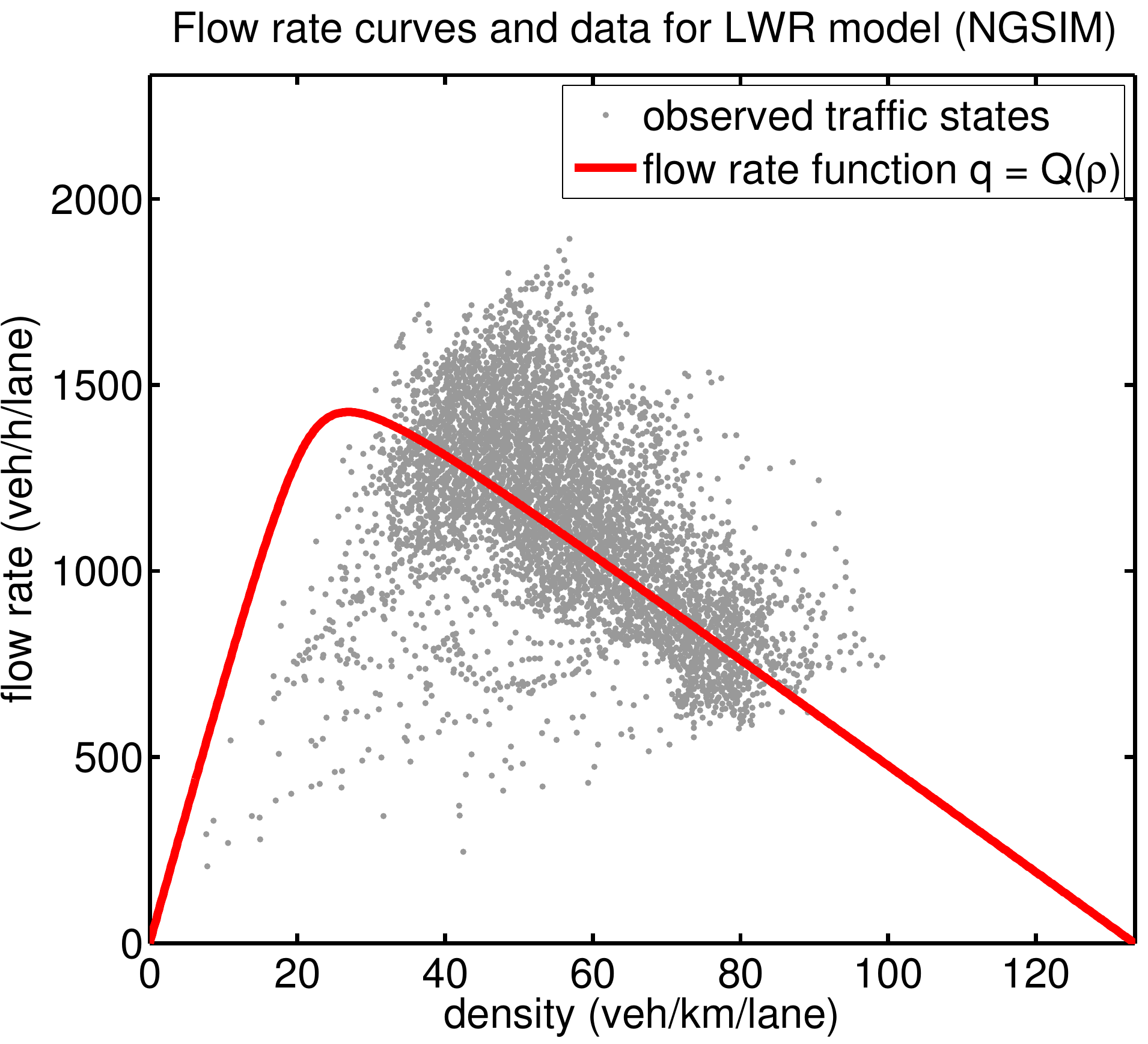}
\end{minipage}
\hfill
\begin{minipage}[b]{.32\textwidth}
\includegraphics[width=\textwidth]{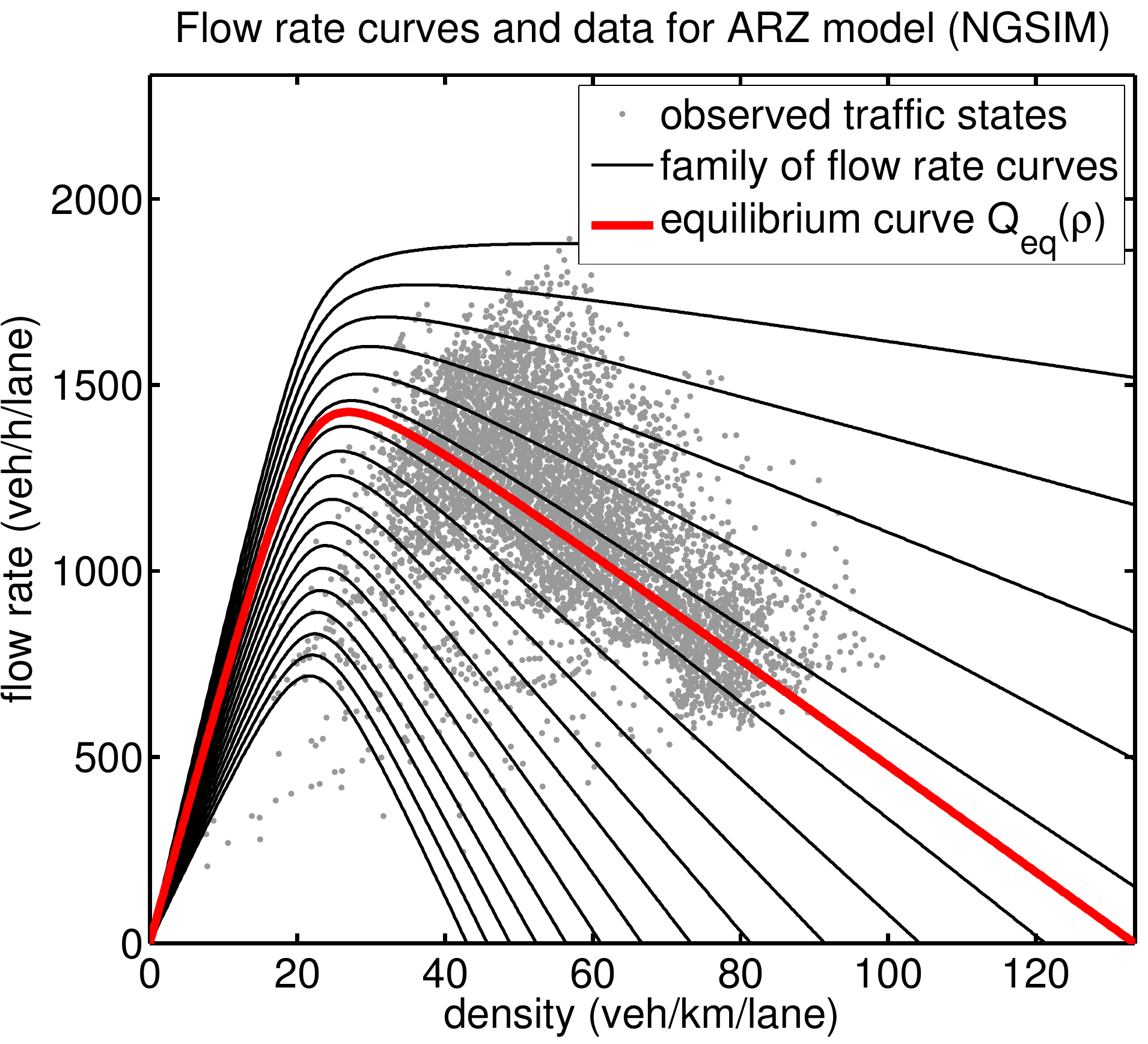}
\end{minipage}
\hfill
\begin{minipage}[b]{.32\textwidth}
\includegraphics[width=\textwidth]{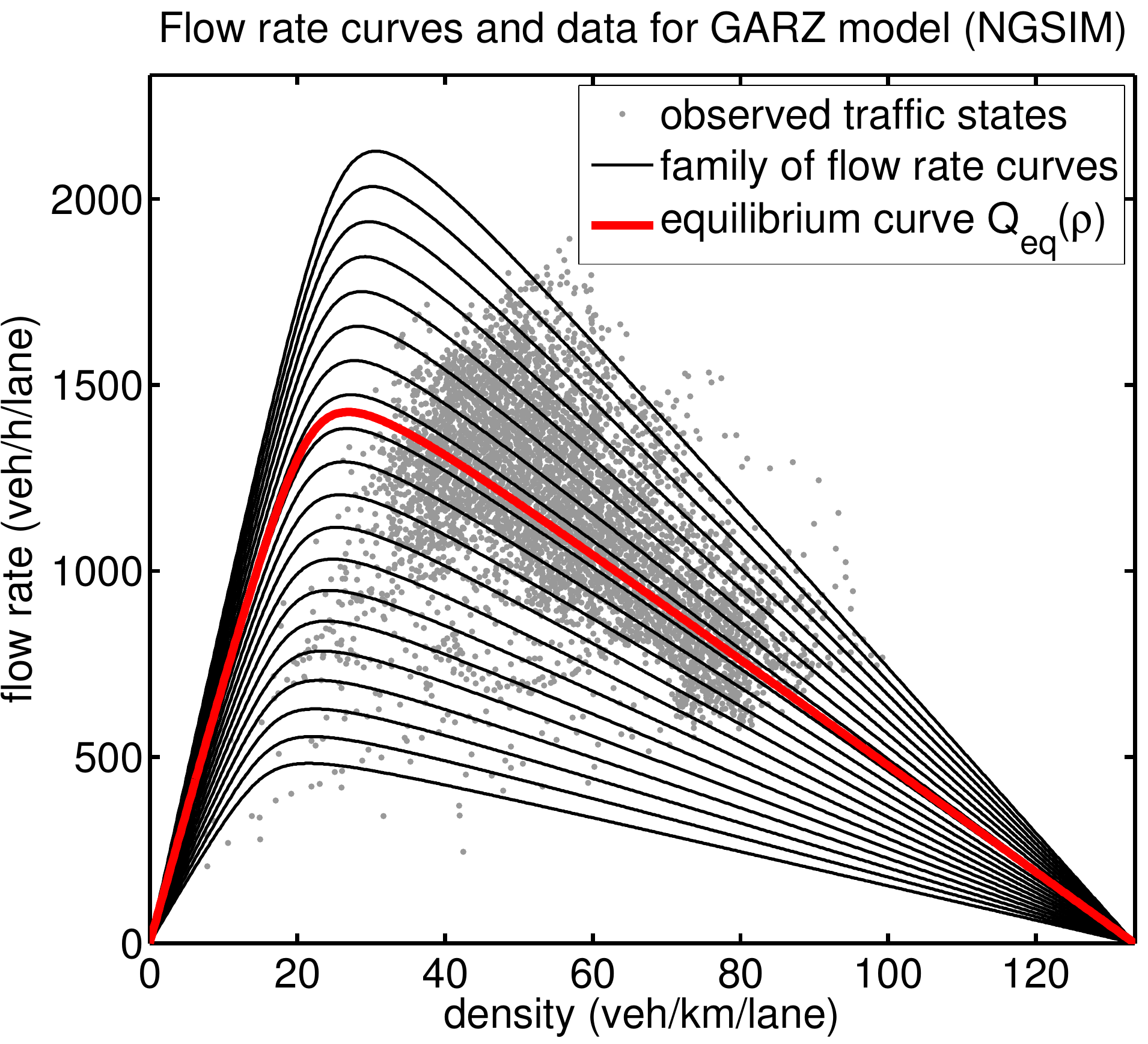}
\end{minipage}

\vspace{.5em}
\begin{minipage}[b]{.32\textwidth}
\includegraphics[width=\textwidth]{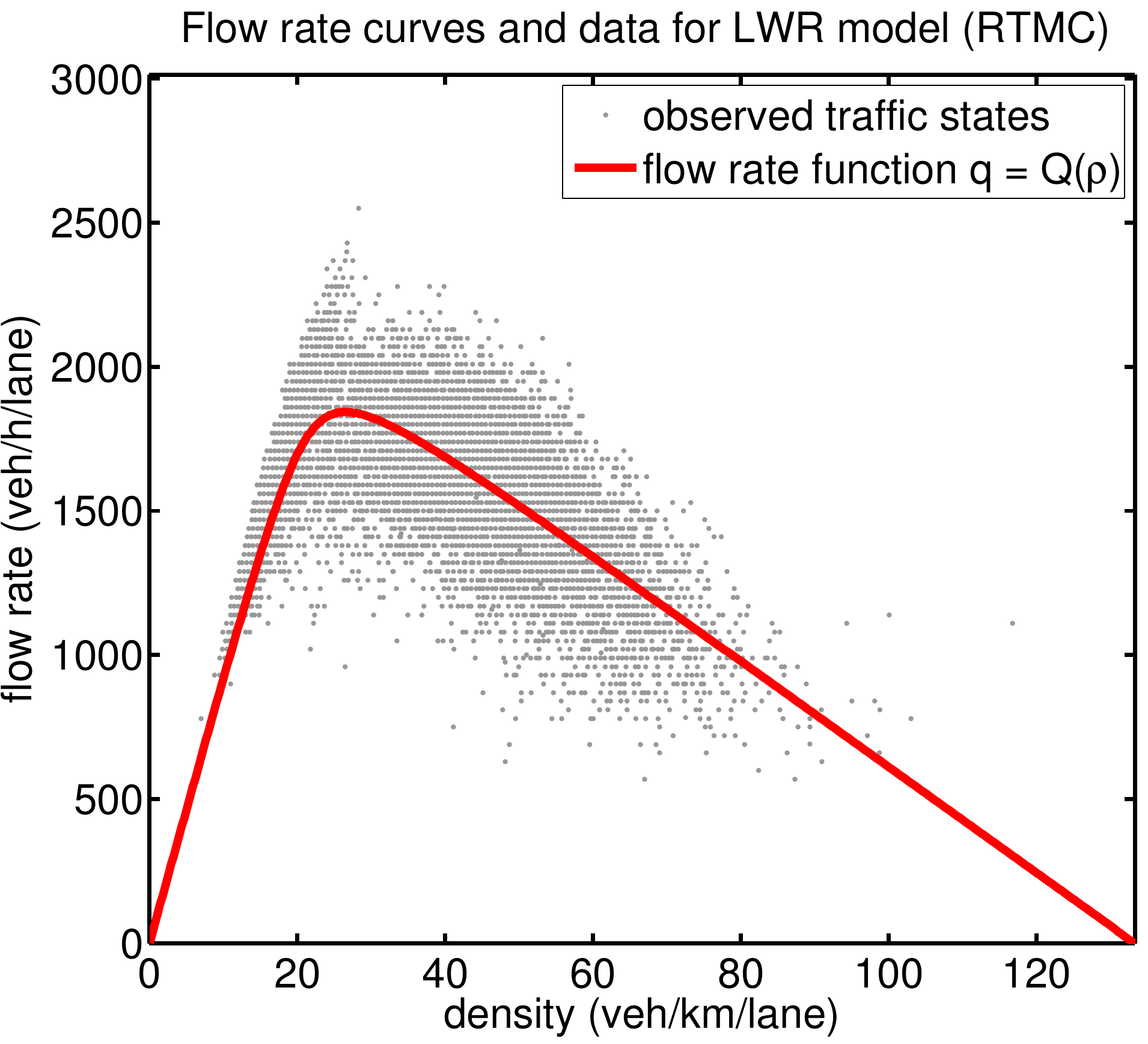}
\end{minipage}
\hfill
\begin{minipage}[b]{.32\textwidth}
\includegraphics[width=\textwidth]{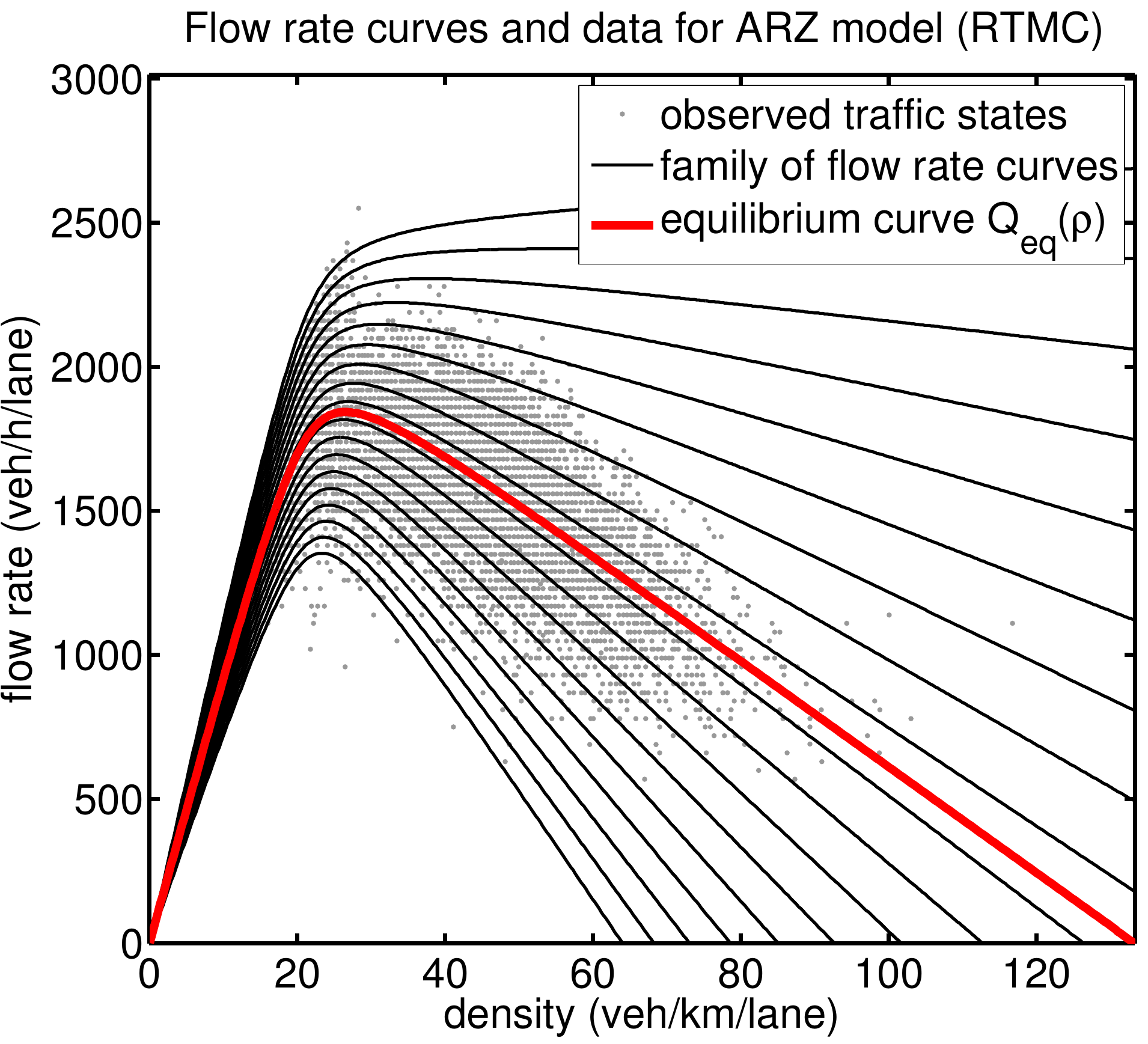}
\end{minipage}
\hfill
\begin{minipage}[b]{.32\textwidth}
\includegraphics[width=\textwidth]{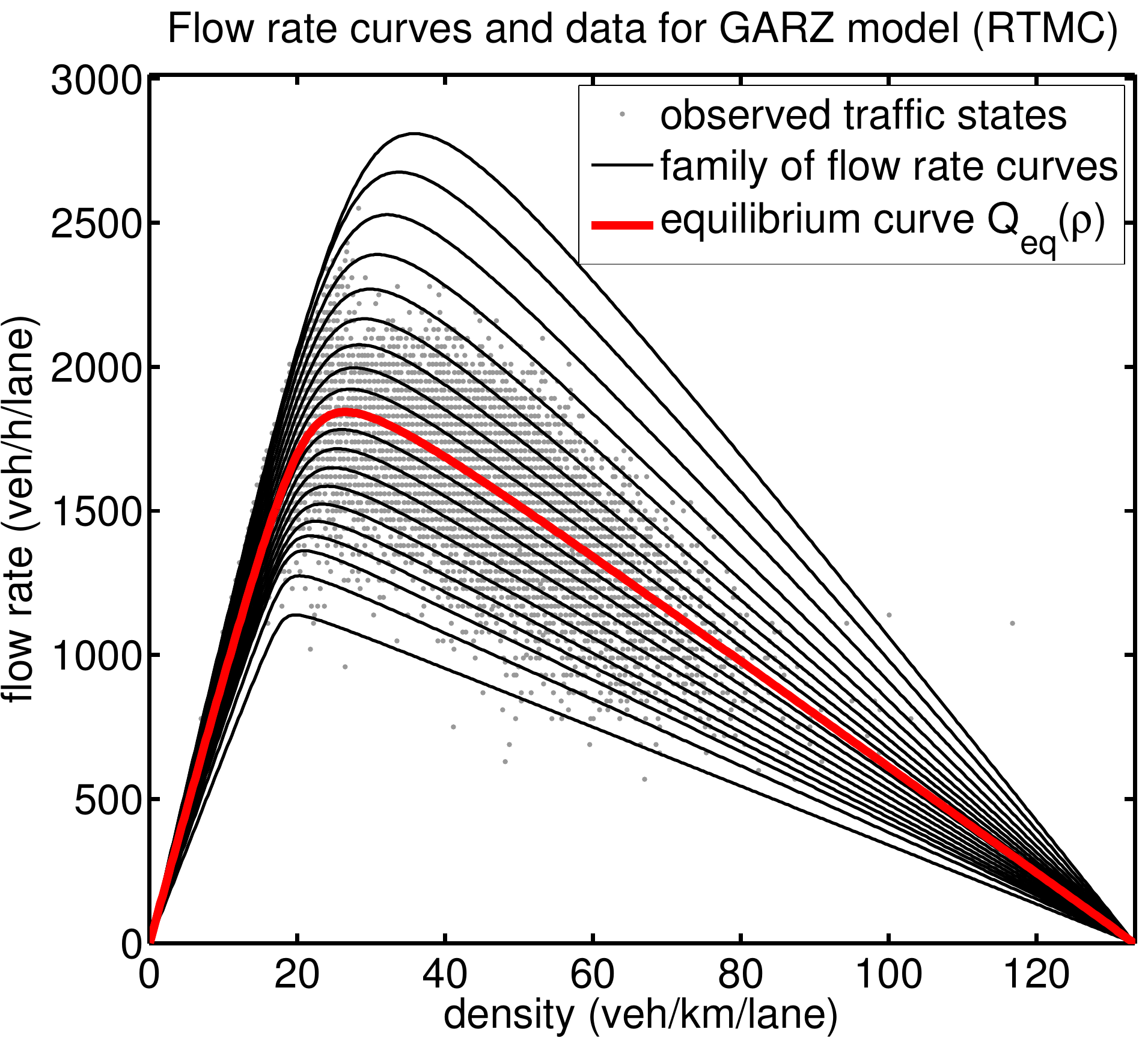}
\end{minipage}

\vspace{-.4em}
\caption{Flow rate vs.\ density curves for the models LWR (left column), ARZ (middle column), and GARZ (right column), together with the traffic states observed in the test cases (gray dots). The results show the NGSIM data (top row) and the RTMC data (bottom row).}
\label{fig:models_flow_rate_curves}
\end{figure}

\subsection{List of Models}
\label{subsec:list_of_models}
We compare the following four models in terms of their predictive accuracy of the real data:
\begin{enumerate}[ 1)]
\item\textbf{Interpolation:} A predictor that, at any instance in time, constructs the traffic density and velocity via direct linear interpolation of the boundary conditions, i.e., on the road segment $x\in [0,L]$, the predicted state is $\rho(x,t) = \rho(0,t)(1-x/L)+\rho(L,t)x/L$ and $u(x,t) = u(0,t)(1-x/L)+u(L,t)x/L$. Of course, this predictor is not an actual traffic model. However, due to its simplistic nature, it represents an important means of comparison.
\item\textbf{LWR:} The first-order model \eqref{eq:lighthill_whitham_richards_model}, in which only the density state $\rho(x,t)$ is evolved, based on the data-fitted equilibrium velocity curve $V(\rho,w_\text{eq})$, resulting from the least-squares fit \eqref{eq:LSQ} of the family \eqref{eq:flow_rate_curve} to the fundamental diagram data.
\item\textbf{ARZ:} The second-order ARZ model \eqref{eq:aw_rascle_zhang_model_w} that generalizes the least-squares fitted flow rate curve of the LWR model to a family of curves $V(\rho,w) = V(\rho,w_\text{eq})+(w-w_\text{eq})$.
\item\textbf{GARZ:} The second-order generalized ARZ model \eqref{eq:GARZ_model}, whose generalized velocity function $V(\rho,w)$ is obtained via a weighted least-squares fit \eqref{eq:WLSQ} of the family \eqref{eq:flow_rate_curve} to the fundamental diagram data.
\end{enumerate}
The fundamental diagram curves $Q_w(\rho) = \rho V(\rho,w)$ of the three traffic models are shown in Fig.~\ref{fig:models_flow_rate_curves}, overlayed on top of the traffic state data that are actually observed in the test cases (gray dots). The top row of figures corresponds to the NGSIM data, and the bottom row represents the RTMC data. The LWR model is shown on the left, the ARZ model in the middle, and the GARZ model on the right.

Regarding the reproduction of real traffic data, it should be stressed that the models/pre{\-}dictors differ in the way they use data. In the model generation step, the Interpolation predictor uses no historic data; the LWR and ARZ model are based on the same least-squares fit; and the GARZ model employs more information from the fundamental diagram data due to the weighted least-squares fit. In turn, during the advance forward in time, the Interpolation predictor uses two pieces of information ($\rho$ and $u$) at each boundary. In contrast, the LWR model uses only one piece of information ($\rho$) at one of the two boundaries (assuming the traffic states at the two boundaries are either both in free flow or both congested). Finally, the ARZ and GARZ model use a total of two pieces of information through the boundary conditions.

\begin{figure}
\begin{minipage}[b]{.495\textwidth}
\includegraphics[width=\textwidth]{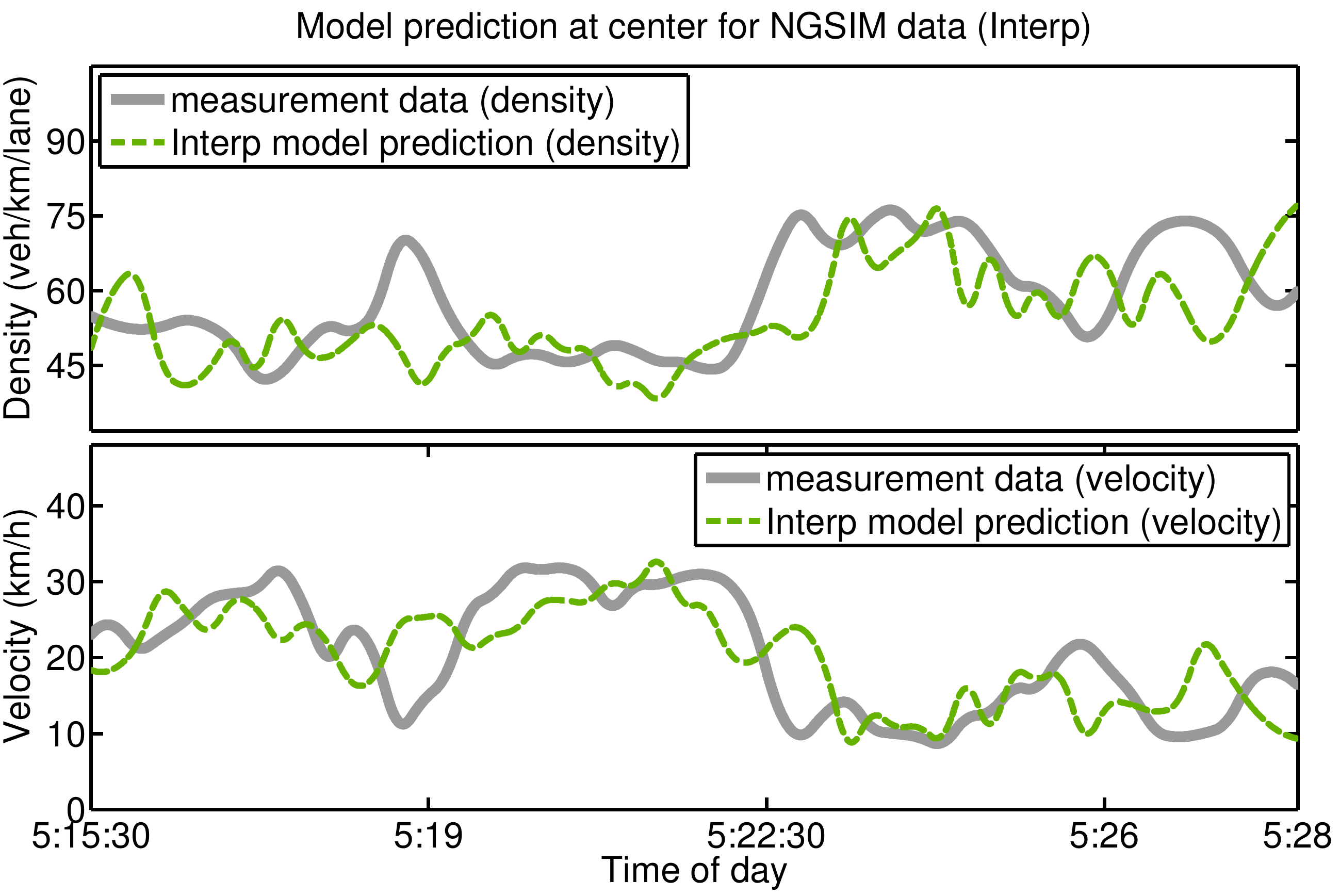}
\end{minipage}
\hfill
\begin{minipage}[b]{.495\textwidth}
\includegraphics[width=\textwidth]{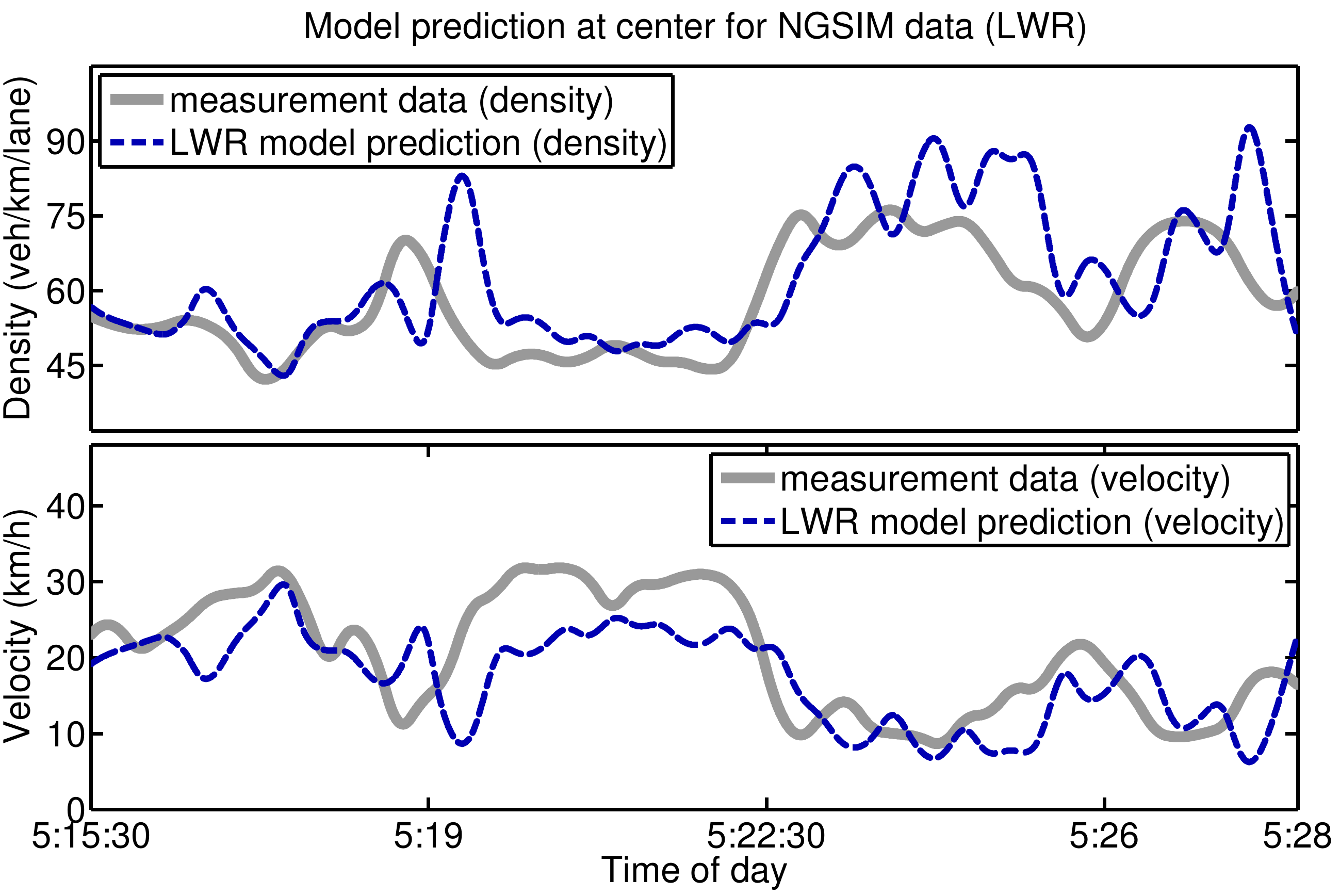}
\end{minipage}

\vspace{.5em}
\begin{minipage}[b]{.495\textwidth}
\includegraphics[width=\textwidth]{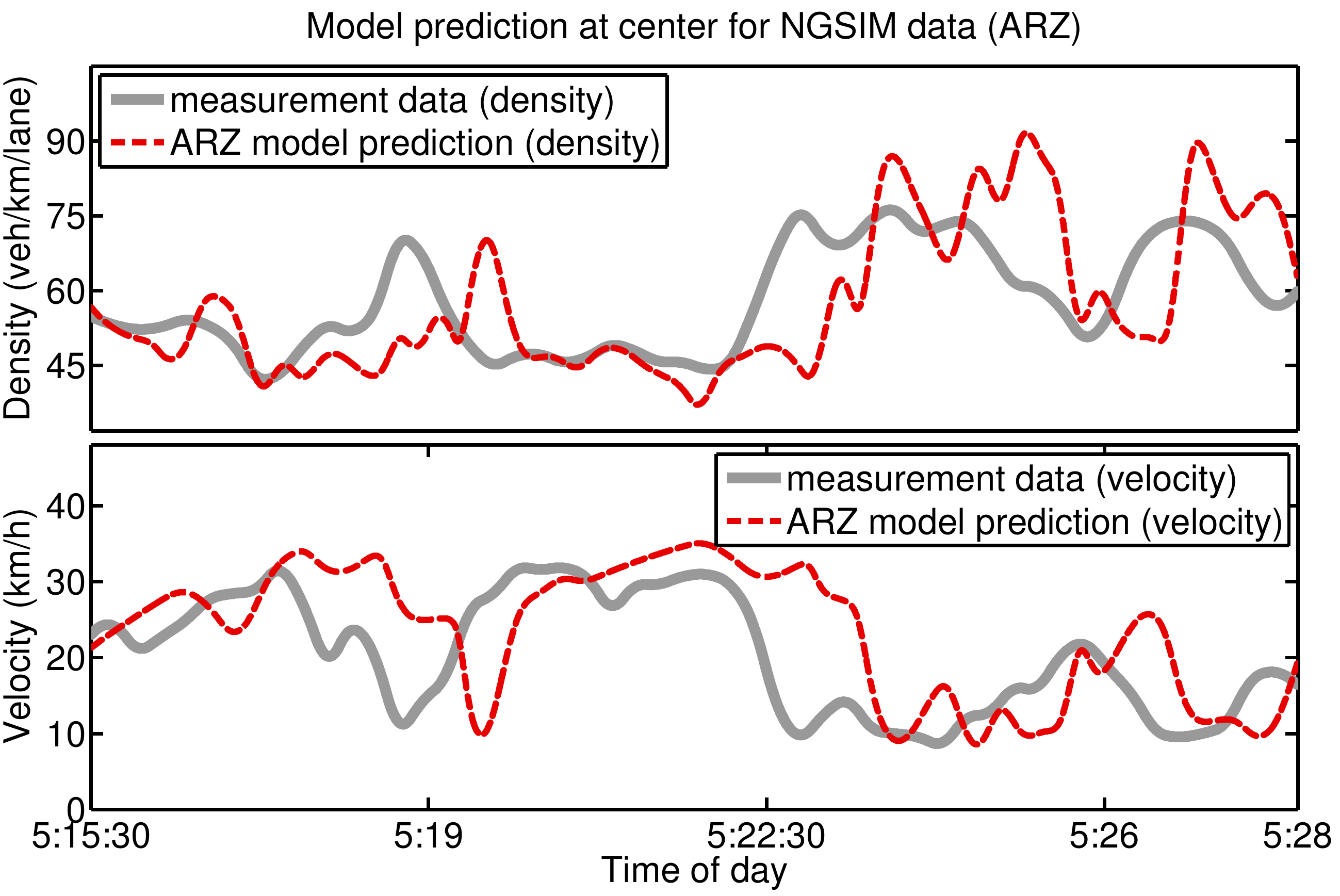}
\end{minipage}
\hfill
\begin{minipage}[b]{.495\textwidth}
\includegraphics[width=\textwidth]{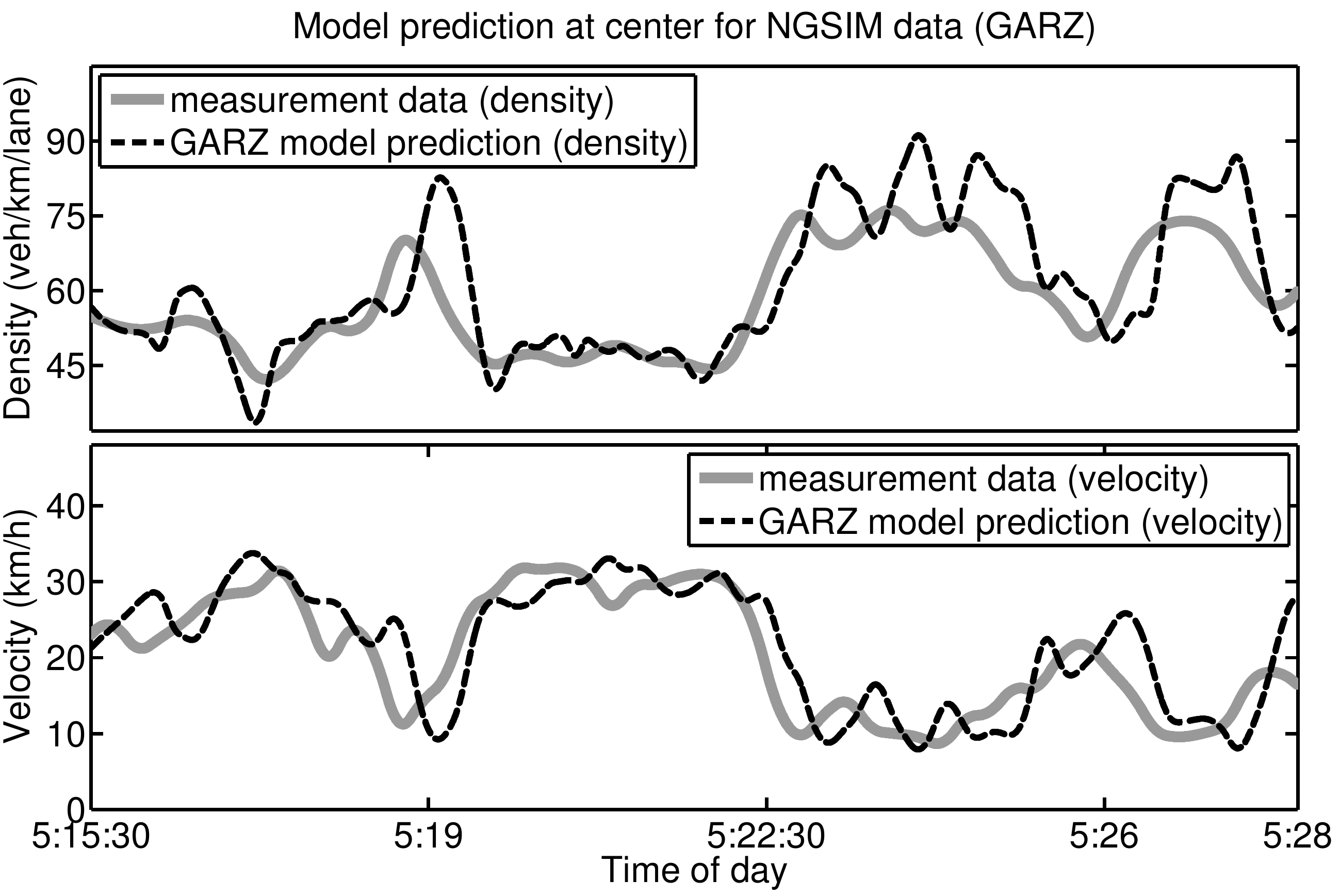}
\end{minipage}

\vspace{-.4em}
\caption{Model comparison on the NGSIM 5:15pm--5:30pm data set. In each panel, the time-evolutions of the model predictions (colored dashed curve) and measured data (solid gray curve) at the middle position of the study area are shown (top box: $\rho$, bottom box: $u$). The four panels correspond to: Interpolation (top left, green), LWR (top right, blue), ARZ (bottom left, red), GARZ (bottom right, black).}
\label{fig:NGSIM_evolution}
\end{figure}

\begin{figure}[p]
\centering
\includegraphics[width=.97\textwidth]{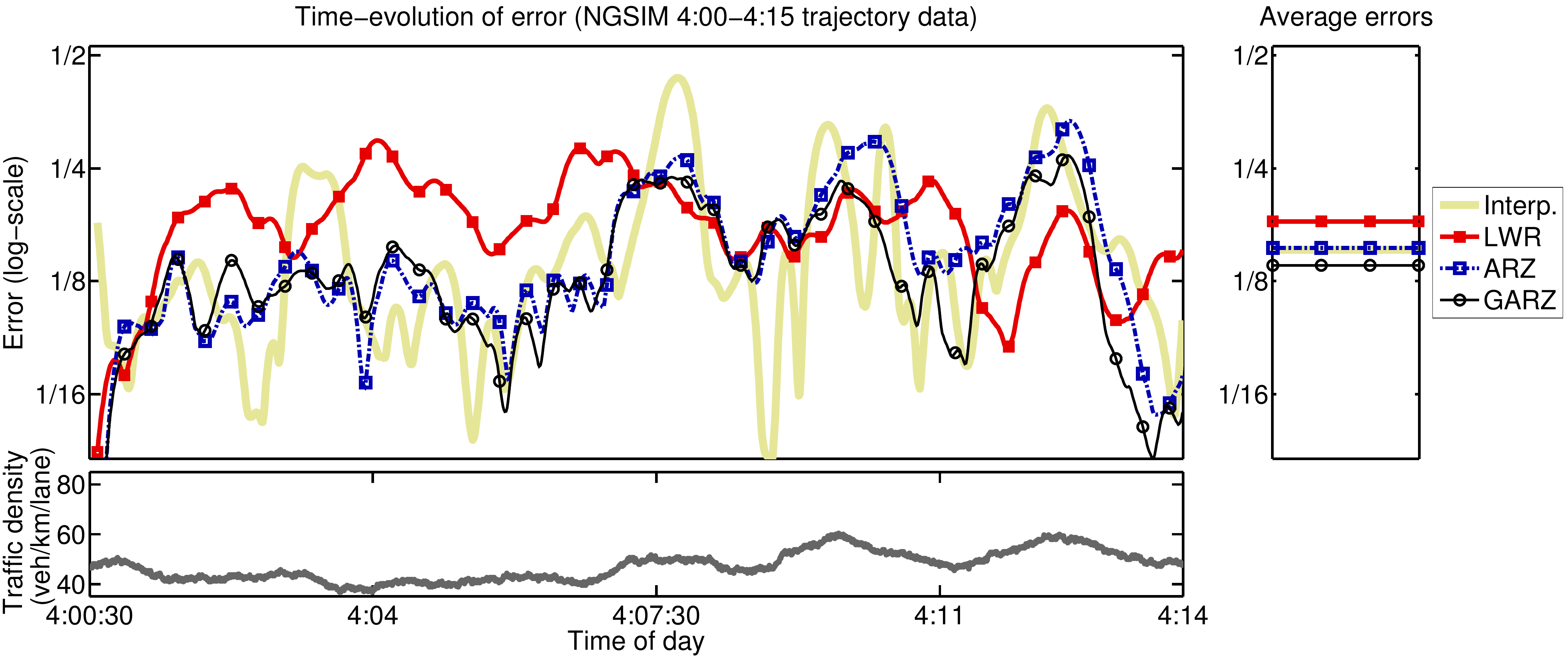}

\vspace{.7em}
\includegraphics[width=.97\textwidth]{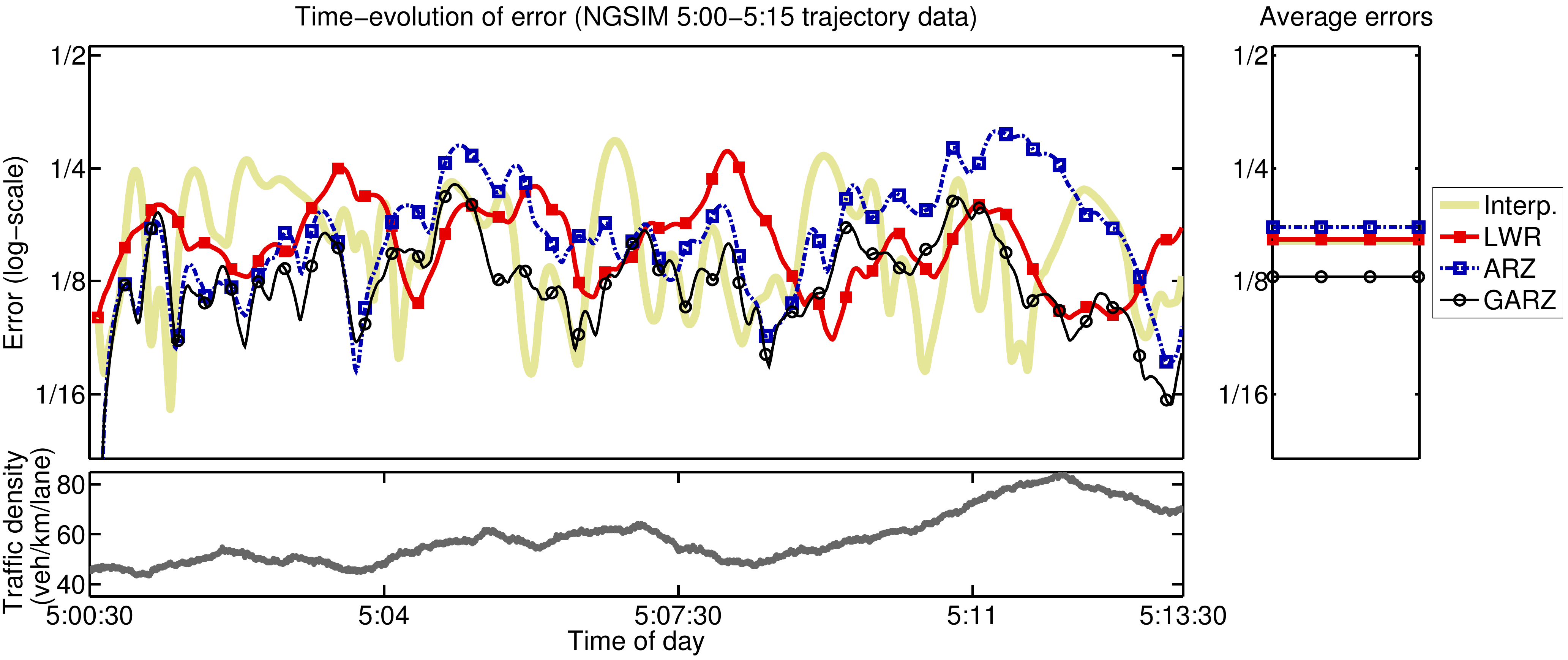}

\vspace{.7em}
\includegraphics[width=.97\textwidth]{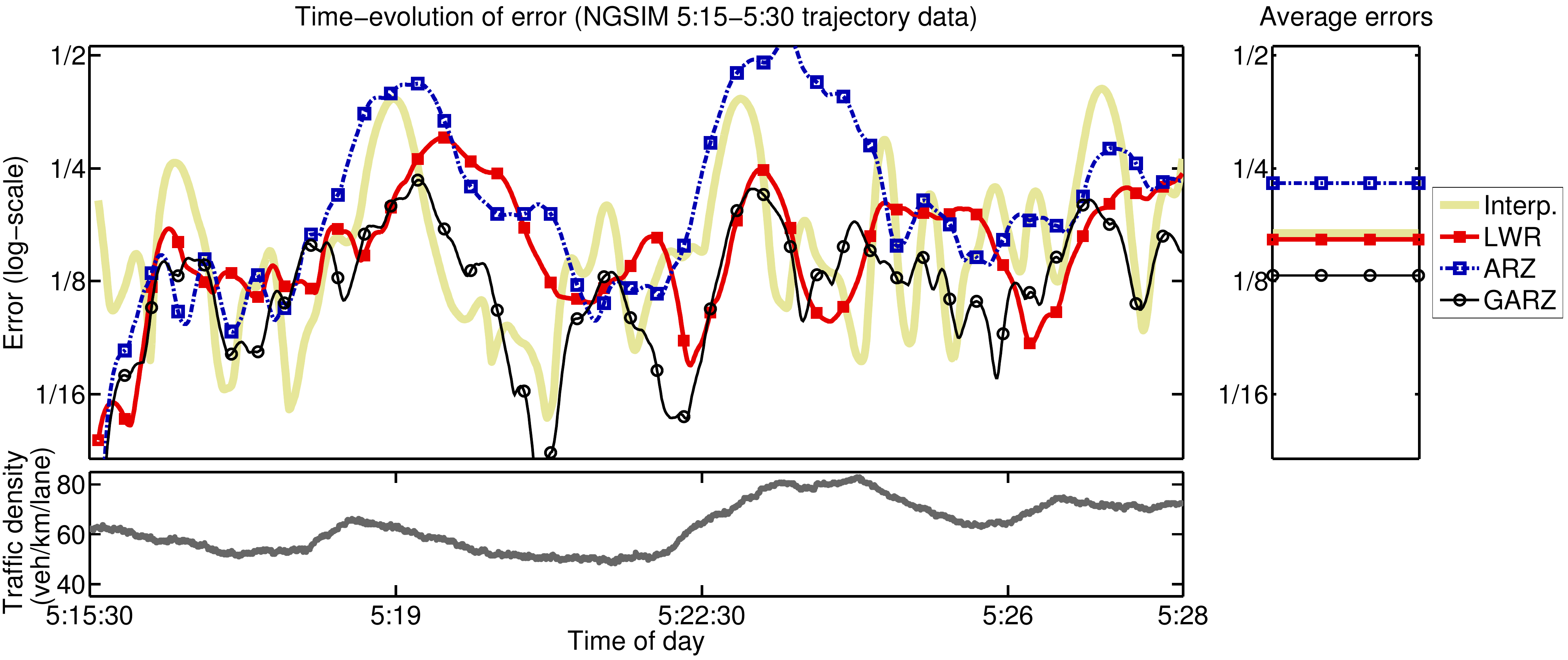}

\vspace{-.4em}
\caption{Comparison of traffic models for the data sets NGSIM 4:00pm--4:15pm (at the top), NGSIM 5:00pm--5:15pm (in the middle), and NGSIM 5:15pm--5:30pm (at the bottom). In each panel, the time-evolution of spatially averaged errors (left top box), measurement traffic density data (left bottom box), and space-time average errors (right box) are shown in log-scale for Interpolation, LWR, ARZ, and GARZ.}
\label{fig:NGSIM_errors}
\end{figure}

\subsection{Model Validation with the NGSIM Trajectory Data}
\label{subsec:validation_ngsim}
As in \cite{FanSeibold2013}, we consider a segment of 450 meters in length in the domain of the NSGIM vehicle trajectories. Data for three time intervals is available: 4:00pm--4:15pm, 5:00pm--5:15pm, and 5:15pm--5:30pm. However, because our model validation requires a traffic state be defined on the complete study segment, slightly shorter study time intervals must be chosen that guarantee that recorded vehicles are present everywhere on the road. We choose: 4:00:30pm--4:14:00pm, 5:00:30pm--5:13:30pm, and 5:15:30pm--5:28:00pm. The parameters of the traffic models, obtained by fitting to the fundamental diagram data provided with NSGIM, are given in the NGSIM row of Table~\ref{tab:model_parameters}. In line with the temporal resolution of the data, we generate density and velocity functions (real data and model predictions) in intervals of 0.1 seconds.

Figure~\ref{fig:NGSIM_evolution} displays the time-evolution (5:15pm--5:30pm) of the traffic density and velocity that are predicted by the selected models, in comparison with the real evolution of these quantities, at the mid-point of the study area. Both the LWR and the ARZ model do reproduce the general trends present in the true density evolution, albeit with a systematic delay of 30--60 seconds (see \cite{FanSeibold2014} for a detailed discussion on this delay). One difference between LWR and ARZ is that the latter reproduces, modulo the delay, the velocity evolution better than the former. This is particularly visible in the predicted velocity in the plateau between 5:19:00pm and 5:22:30pm. In turn, the GARZ model captures the evolution of density and velocity significantly better than the other two traffic models. There is still a systematic delay in the model predictions, but this delay is very small. Finally, the Interpolation predictor captures the large--scale features of the real data as well. However, it also exhibits an oscillatory behavior of a larger frequency. This is due to the fact that the linear interpolations transmit temporal oscillations from both boundaries immediately to the observation site. In contrast, the traffic models pick up fewer boundary data and furthermore oscillations can turn into N-waves and thus reduce in magnitude as they move into the domain.

\begin{table}
\begin{tabular}{|c||c|c|c|c|c|}
\hline
Data set & $\alpha$ (veh/h/lane) & $\lambda$ & $p$ & $w_\text{min}$ (km/h) & $w_{\text{max}}$ (km/h) \\
\hline\hline
NGSIM & 247.38 & 23.41 & 0.16 & 39.49 &  82.01 \\
RTMC  & 316.46 & 23.91 & 0.16 & 75.81 & 102.82 \\
\hline
\end{tabular}
\vspace{.6em}
\caption{Model parameters for the data-fitted models for the two data sets. Here, $\alpha$, $\lambda$, and $p$ are the free parameters of the equilibrium flow rate curve $Q_\text{eq}(\rho)$ in the family \eqref{eq:flow_rate_curve}. Moreover, $w_\text{min}$ and $w_\text{max}$ denote the boundaries of the empty road velocity for the GARZ model, as described in \S\ref{sec:data-fitting_GARZ}.}
\label{tab:model_parameters}
\end{table}

To quantify the predictive accuracy of the different models, we turn to the error metric \eqref{eq:error_measure}. Figure~\ref{fig:NGSIM_errors} shows the model errors for the NGSIM data sets: 4:00pm--4:15pm (top), 5:00pm--5:15pm (middle), and 5:15pm--5:30pm (bottom). In each panel, the time-evolution of the spatially averaged error \eqref{eq:error_x} is shown (top-left box), as well as the spatio-temporal average error \eqref{eq:error_xt} (top-right box). The models are: Interpolation (thick solid yellow), LWR (solid red), ARZ (dashed blue), and GARZ (thin solid black). In each bottom-left box, the time-evolution of the average traffic density is shown. The numerical values of the model errors \eqref{eq:error_xt} are given in rows 2--4 of Table~\ref{tab:model_errors}.

\begin{table}
\begin{tabular}{|ll||lr|lr|lr|lr|}
\hline
Data set\rule{0em}{1.0em}\hspace{-1em} &&
Interp. \hspace{-10em}& &LWR\hspace{-10em} & & ARZ\hspace{-10em} & & GARZ\hspace{-5em} &  \\
\hline\hline
NGSIM\rule{0em}{0em}
      & 4:00--4:15 &0.151 &(+10\%) & 0.181 &(+31\%) & 0.153 &(+11\%) & 0.138 &\\
NGSIM & 5:00--5:15 &0.160 &(+25\%) & 0.161 &(+26\%) & 0.174 &(+35\%) & 0.129 &\\
NGSIM & 5:15--5:30 &0.168 &(+30\%) & 0.162 &(+25\%) & 0.228 &(+76\%) & 0.129 &\\
\hline
RTMC\rule{0em}{0em}
     & congested   &0.203 &(+14\%) & 0.228 &(+24\%) & 0.208 &(+13\%) & 0.184 & \\
RTMC & non-cong.   &0.108 &(+63\%) & 0.081 &(+26\%) & 0.067 & (+4\%) & 0.064 & \\
\hline
\end{tabular}
\vspace{.6em}
\caption{Spatio-temporal average errors of the traffic models and the Interpolation predictor for the NGSIM data sets (rows 2--4) and the RTMC data (rows 5--6), separated into congested and non-congested days. In each row, the parentheses denote the excess error relative to the best model (GARZ in all cases).}
\label{tab:model_errors}
\end{table}

The time-evolution of the errors \eqref{eq:error_x} confirm several of the above observations, such as the highly oscillatory nature of the Interpolation predictor, and the good accuracy of the GARZ model. Another particularly visible effect is the bad performance of the ARZ model for large densities: for 4:00--4:15, the ARZ model yields smaller errors than LWR; in contrast, for 5:15--5:30, the ARZ leads to larger errors than LWR. Moreover, the peaks in the ARZ errors coincide with local maxima in the traffic density. These observations confirm that: a) the ARZ model has the potential to improve upon the LWR model, and b) the non-uniform stagnation density of the ARZ model significantly affects its predictive accuracy. In contrast, the GARZ model inherits the advantages of the ARZ model for low densities, and furthermore remedies its shortcoming for high densities.

Regarding the performance of the Interpolation predictor; its accuracy is, except for the large oscillations, surprisingly good. Specifically, its average errors \eqref{eq:error_xt} are similar to those of the LWR and ARZ models. The following explanations can be provided to address this observation:
\begin{enumerate}[ a)]
\item The considered road segment is very short. Hence, the coherence between the boundaries and the inside of the domain is high, and the actual delays due to finite-speed information propagation are small. One can expect that on longer road segments, the finite-speed transport of information becomes more relevant, and therefore interpolation significantly loses in accuracy.
\item As described in \S\ref{subsec:list_of_models}, the Interpolation predictor picks up twice (four times) as much data from the boundaries as the second (first) order models. One could therefore argue that the traffic models are as accurate as the interpolation, while utilizing less input data.
\item The Interpolation predictor is on par with an ``incomplete'' model (LWR does not evolve velocities) and a ``defective'' model (ARZ is flawed for large densities). In turn, it does not achieve the accuracy of the GARZ model.
\end{enumerate}

\begin{figure}
\begin{minipage}[b]{.495\textwidth}
\includegraphics[width=\textwidth]{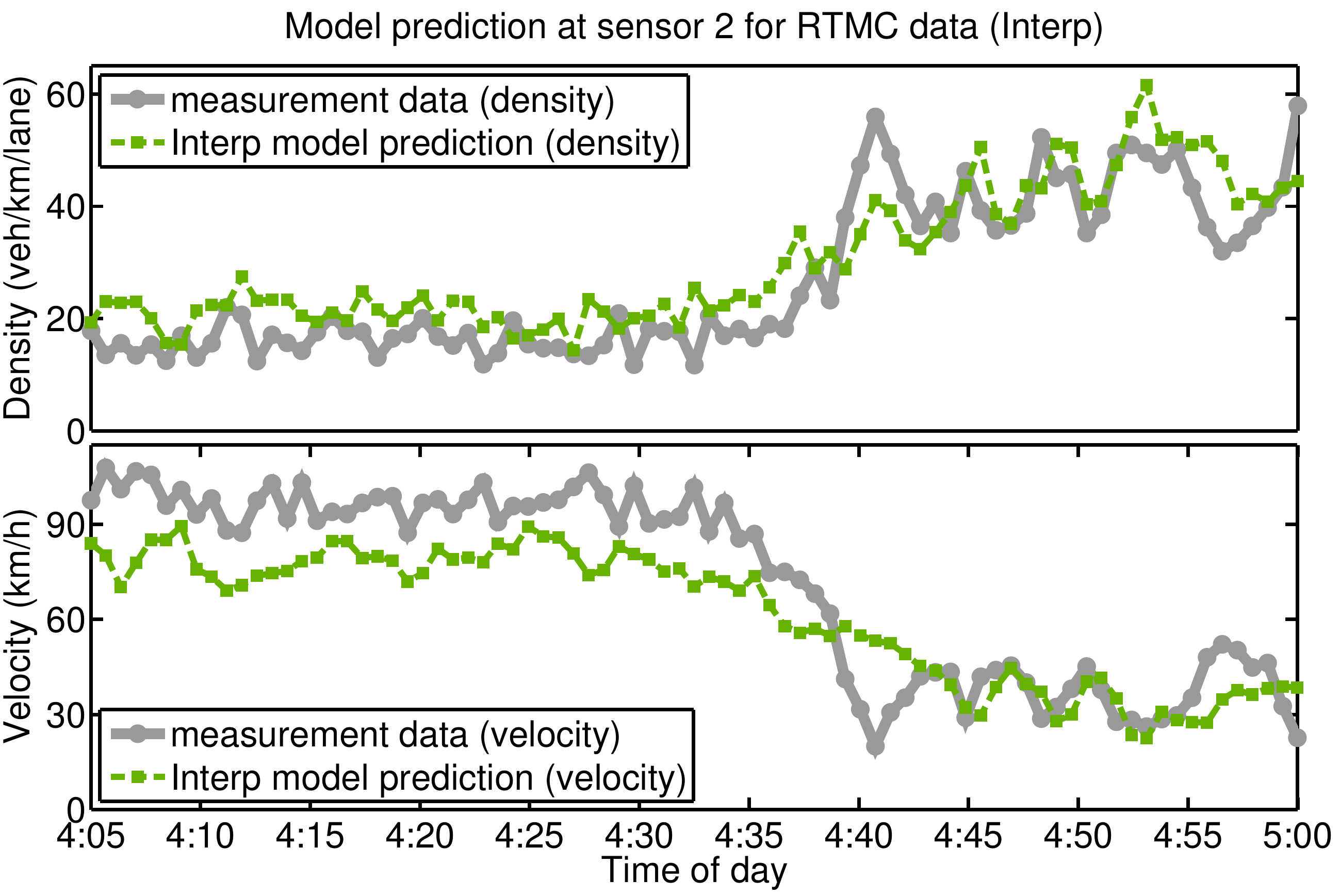}
\end{minipage}
\hfill
\begin{minipage}[b]{.495\textwidth}
\includegraphics[width=\textwidth]{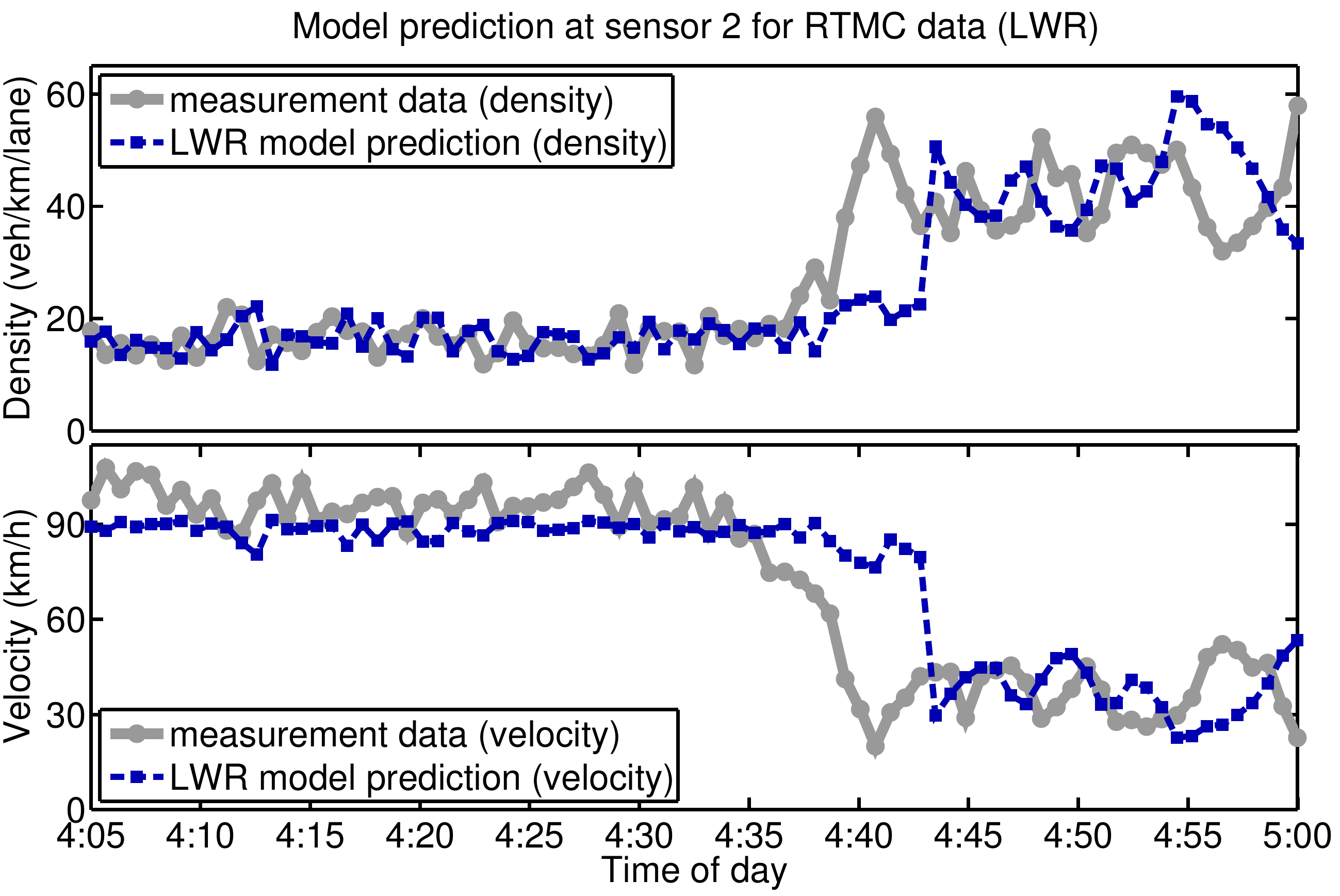}
\end{minipage}

\vspace{.5em}
\begin{minipage}[b]{.495\textwidth}
\includegraphics[width=\textwidth]{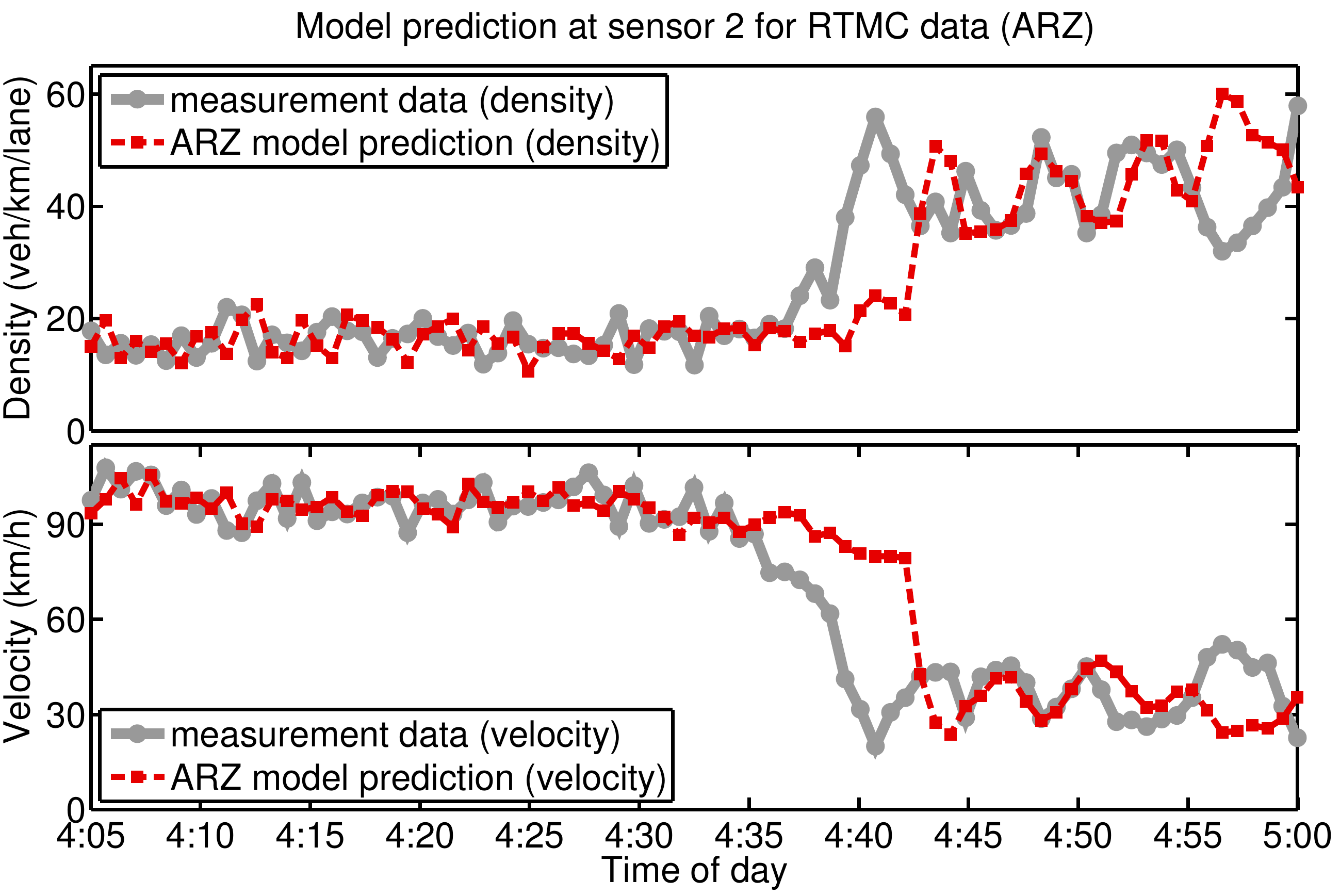}
\end{minipage}
\hfill
\begin{minipage}[b]{.495\textwidth}
\includegraphics[width=\textwidth]{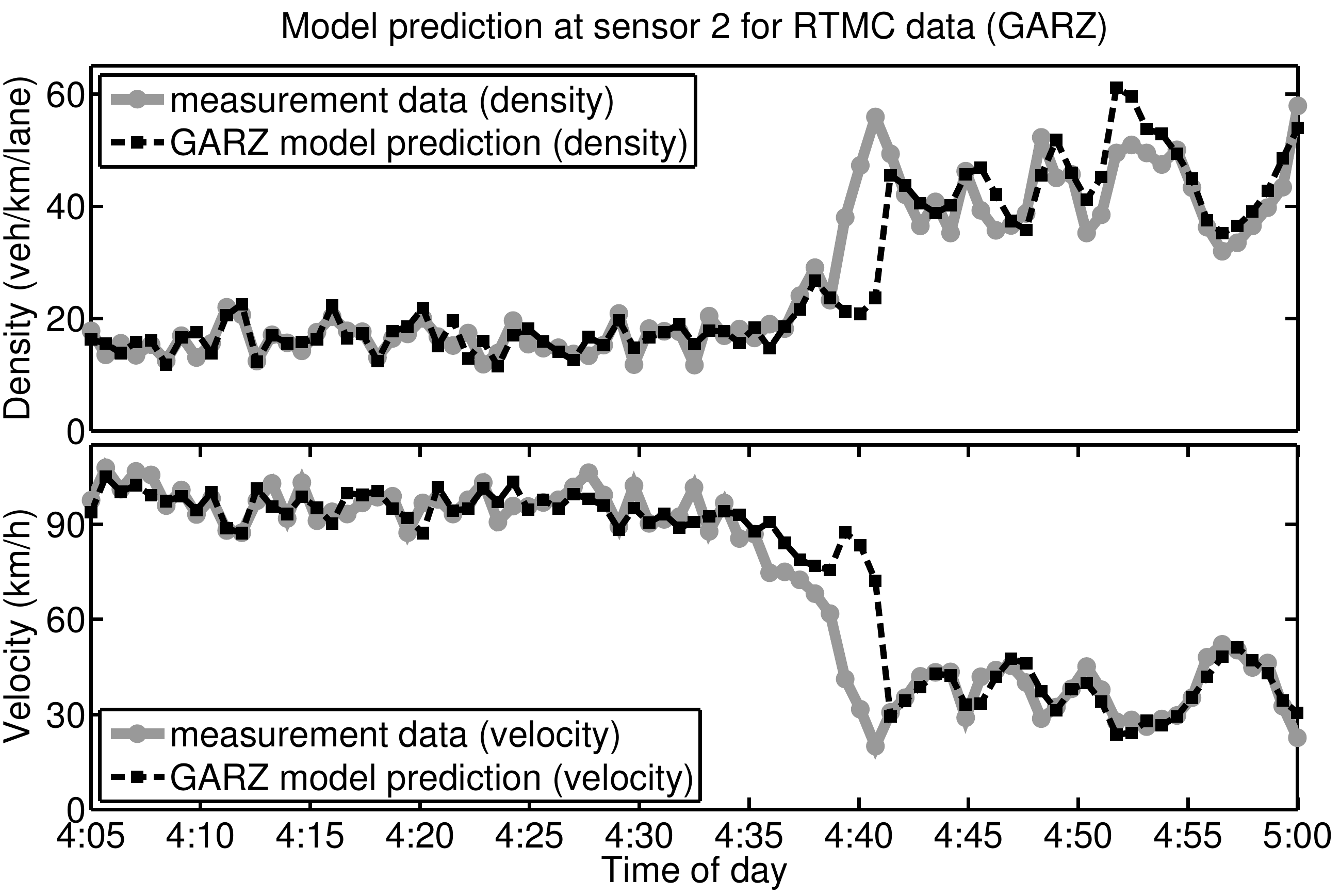}
\end{minipage}

\vspace{-.4em}
\caption{Model comparison on March 26, 2003 in the RTMC data set. In each panel, the time-evolutions (4:05pm--5:00pm) of the model predictions (colored dashed curve) and measured data (solid gray curve) at the middle sensor are shown (top box: $\rho$, bottom box: $u$). The four panels correspond to: Interpolation (top left, green), LWR (top right, blue), ARZ (bottom left, red), GARZ (bottom right, black).}
\label{fig:RTMC_evolution}
\end{figure}

\begin{figure}
\centering
\includegraphics[width=.97\textwidth]{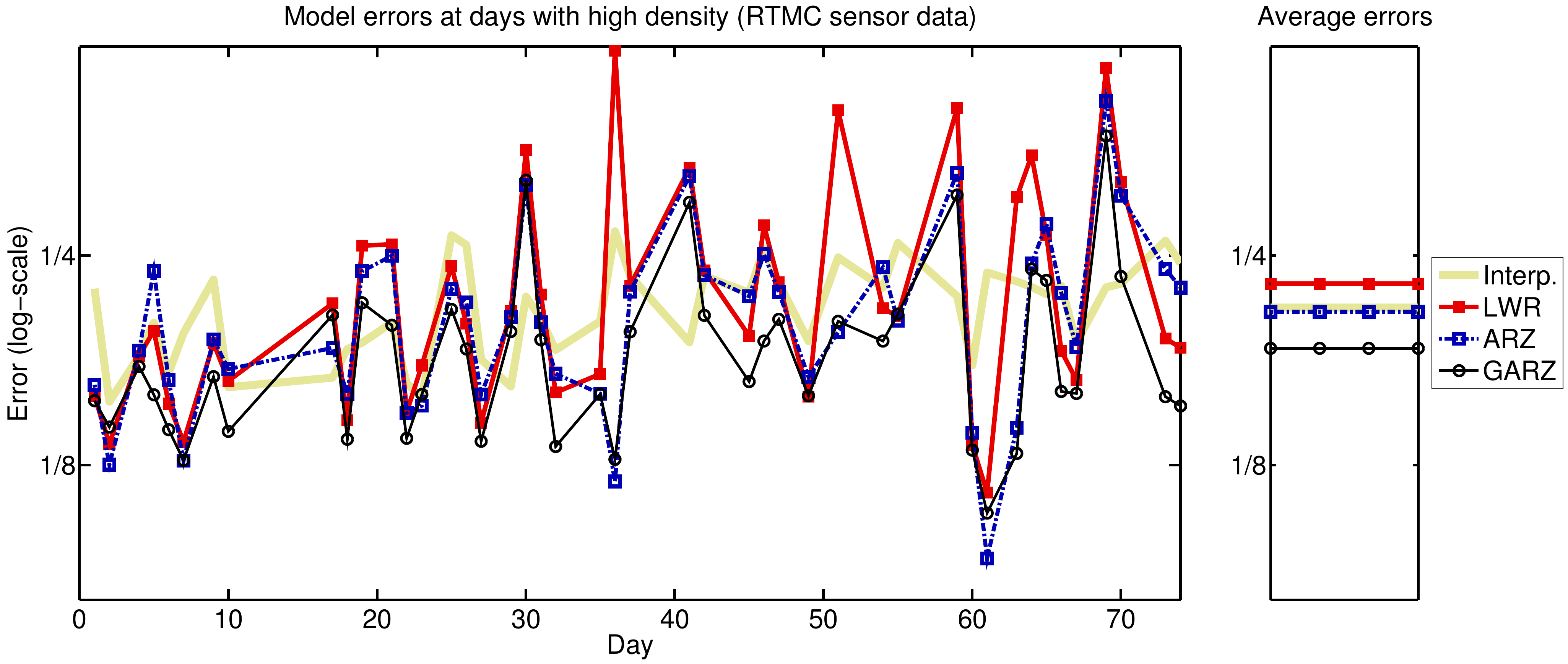}

\vspace{.7em}
\includegraphics[width=.97\textwidth]{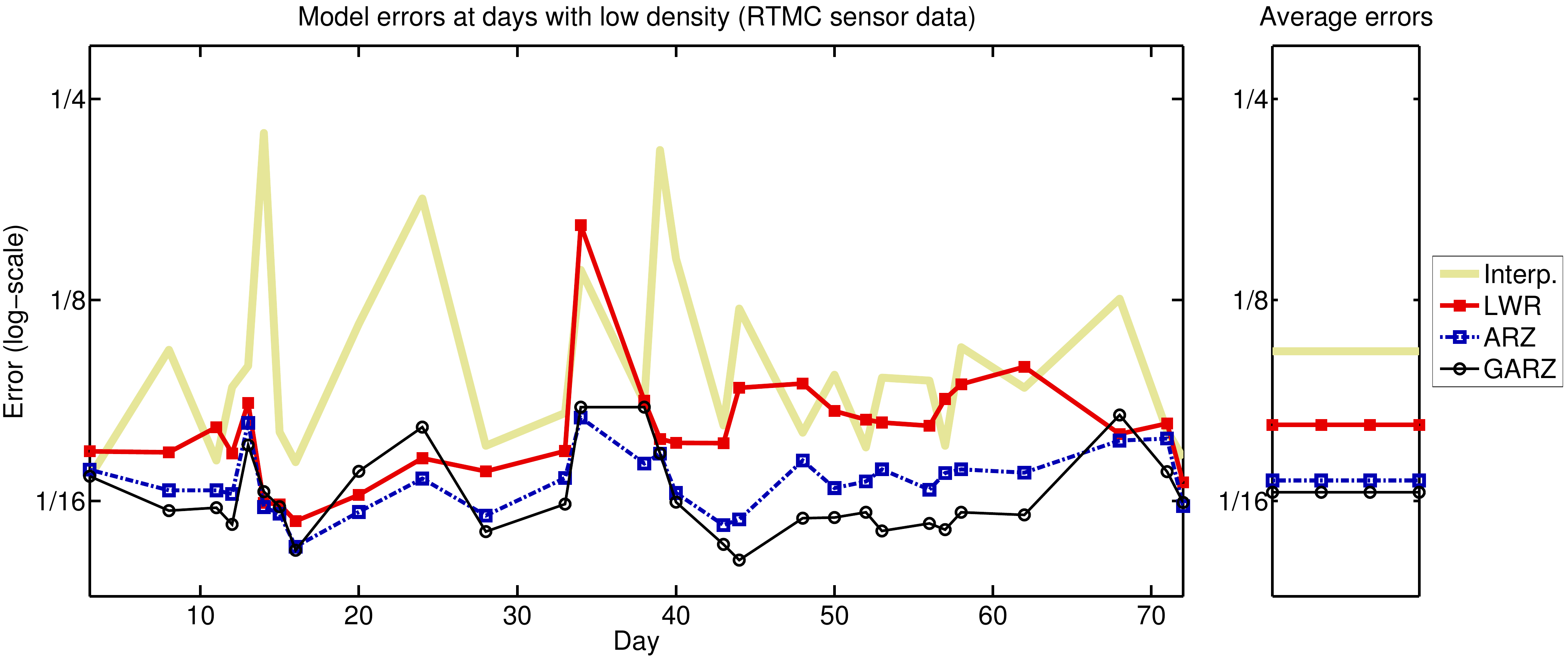}

\vspace{-.4em}
\caption{Comparison of models for the RTMC sensor data on 45 week days with congested traffic (top) and 29 days with non-congested traffic (bottom). The left boxes show the time-averaged error for each day \eqref{eq:error_t}, and the right boxes show the fully averaged errors \eqref{eq:error_day} of congested and non-congested days. The errors for Interpolation, LWR, ARZ, and GARZ are depicted in log-scale.}
\label{fig:RTMC_errors}
\end{figure}

\subsection{Model Validation with the RTMC Sensor Data}
In line with the test described in \cite{FanSeibold2013}, we consider a 1.214\,km long segment of highway on the I-35W, Minneapolis. Two traffic sensors are at the ends of the study segment, and one sensor is inside of the segment (roughly in the middle). The parameters of the data-fitted models, obtained from one-year data at the middle sensor, are given in the RTMC row of Table~\ref{tab:model_parameters}. On each weekday (Monday--Friday) between 01/01/2003 and 04/14/2003, we consider the onset of afternoon rush hour between 4pm and 5pm, and we divide the 74 days into 45 days with congested traffic (the space-time averaged traffic density exceeds 20\,veh/km/lane) and 29 days with non-congested traffic (the remaining ones). As described in \S\ref{subsec:data_treatment}, the traffic models are run through an initialization phase 4:00pm--4:05pm, in which boundary data create a realistic state inside the segment. The actual validation is then conducted in 4:05pm--5:00pm.

Analogously as for the NGSIM data, we first consider the temporal evolution of the traffic states that the models predict and study the qualitative behavior of the models. We look at a day (03/26/2003) on which congestion builds up rapidly (between 4:35pm and 4:40pm) at the sensor position, so that any delays in the model behavior become visible. The results, shown in Fig.~\ref{fig:RTMC_evolution}, confirm the observations made for the NGSIM data: a) LWR and ARZ predict similar densities, but ARZ does a better job at also capturing velocities correctly; b) both LWR and ARZ propagate information too slowly, resulting in the onset of congestion to be predicted 5 minutes too late; c) the GARZ model is not perfect (it still predicts the onset of congestion 2 minutes too late), but it captures the general trends in both variables quite nicely. This last point is particularly visible in the velocity profile in the congested state in 4:42pm--5:00pm.

One aspect that is different from the NGSIM test is that the Interpolation predictor performs visibly worse than the traffic models during the low density and high velocity period 4:05pm--4:35pm. This demonstrates the assertion made in \S\ref{subsec:validation_ngsim}, that simple interpolation performs much worse on longer road segments.

The average model errors obtained with the RTMC data are shown in Fig.~\ref{fig:RTMC_errors}. The top panel collects the 45 congested days, and the bottom panel contains the 29 non-congested days. In the left boxes the time-averaged errors \eqref{eq:error_t} in 4:05pm--5:00pm are shown for each day. The right boxes show the resulting total errors \eqref{eq:error_day}, resulting from averaging over all study days. The numerical values, as well as the excess errors of the models relative to GARZ, are shown in rows 5--6 of Table~\ref{tab:model_errors}.

The results of the low density days demonstrate quite clearly that a) traffic models yield noticeably better predictions than simple interpolation; and b) second-order models reproduce the real traffic behavior better than the first-order LWR model. The fact that the GARZ model does not differ much from the ARZ model results from the fact that for low densities, the two families of flow rate curves are very similar (see Fig.~\ref{fig:models_flow_rate_curves}). In turn, the results of the high density days confirm that a) as expected, the quality of the ARZ model worsens relative to the other models; while b) the GARZ model does not suffer from the same amount of accuracy deterioration than ARZ. A somewhat unexpected observation is that for the high-density days, the Interpolation predictor does not perform worse than LWR and ARZ. The reason for this lies is the fact that LWR and ARZ propagate information too slowly, and thus capture features with a delay (see Fig.~\ref{fig:RTMC_evolution}). In contrast, interpolation propagates information instantaneously---which is obviously unrealistic, but here happens to lead to less severe errors than the spurious delays in LWR and ARZ.

\vspace{1.5em}
\section{Inhomogeneous ARZ and GARZ Models}
\label{sec:inhomogeneous_models}
Thus far we have restricted our attention to homogeneous second-order models, in which each vehicle remains for all times attached to its specific velocity curve $u_w(\rho) = V(\rho,w)$. However, it is plausible that in real traffic, drivers vary their empty road velocity $w$ over time, and that the overall traffic dynamics tend towards an equilibrium velocity curve $U_\text{eq}(\rho) = V(\rho,w_\text{eq})$. We therefore extend our model validations to the inhomogeneous ARZ model \eqref{eq:aw_rascle_zhang_model_w_inhomogeneous}, denoted ARZ-$\tau$, and the inhomogeneous GARZ model \eqref{eq:GARZ_model_inhomogeneous}, denoted GARZ-$\tau$. In both cases the equilibrium velocity curve is the LWR velocity curve.

As argued in \S\ref{subsubsec:relaxtion_GARZ_LWR}, for a Cauchy problem the solutions of the inhomogeneous second-order models converge (as $t\to\infty$ for $\tau$ fixed; or as $\tau\to 0$ for $t_\text{final}$ fixed) to solutions of the first-order LWR model. In the presence of boundary data, the same statement holds; however, the limits of the second-order model solutions are LWR solutions \emph{with different boundary conditions} than the LWR solutions that we consider here. Consequently, the models ARZ-$\tau$ and GARZ-$\tau$ are not merely perturbations of the LWR model, but instead could reproduce the dynamics of real traffic better than LWR and better than homogeneous second-order models.

\begin{figure}
\begin{minipage}[b]{.49\textwidth}
\includegraphics[width=\textwidth]{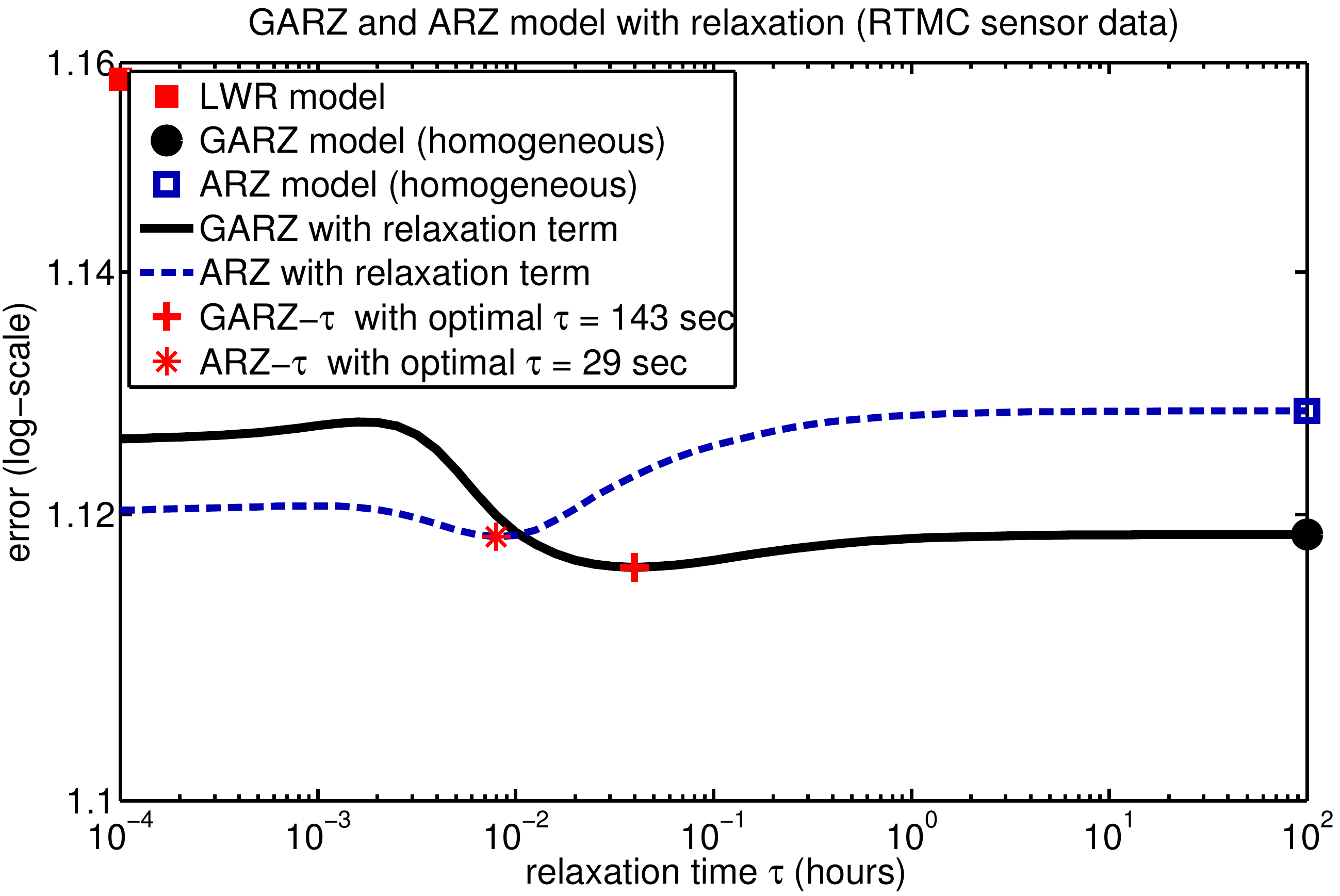}
\end{minipage}
\hfill
\begin{minipage}[b]{.49\textwidth}
\includegraphics[width=\textwidth]{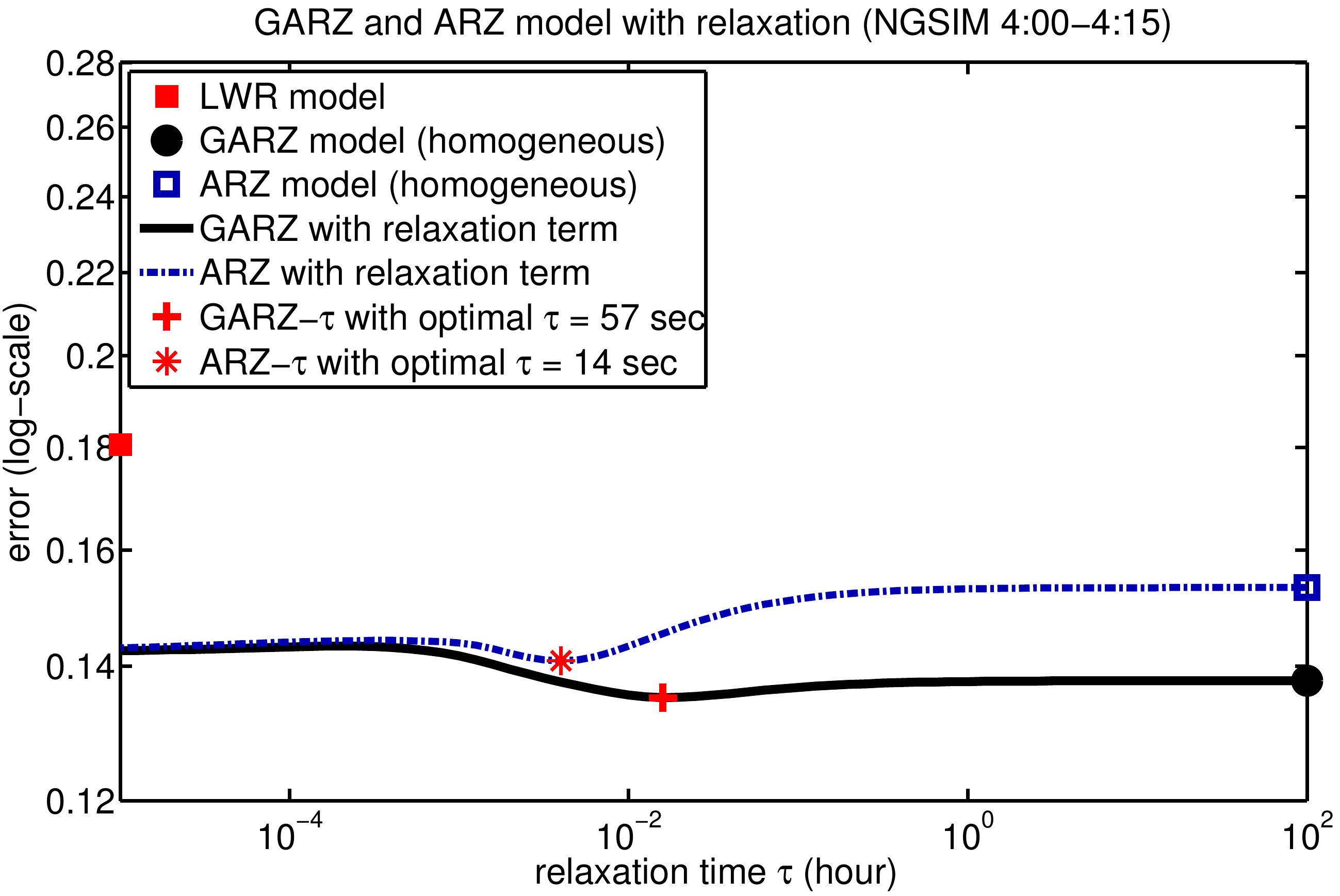}
\end{minipage}

\vspace{.5em}
\begin{minipage}[b]{.49\textwidth}
\includegraphics[width=\textwidth]{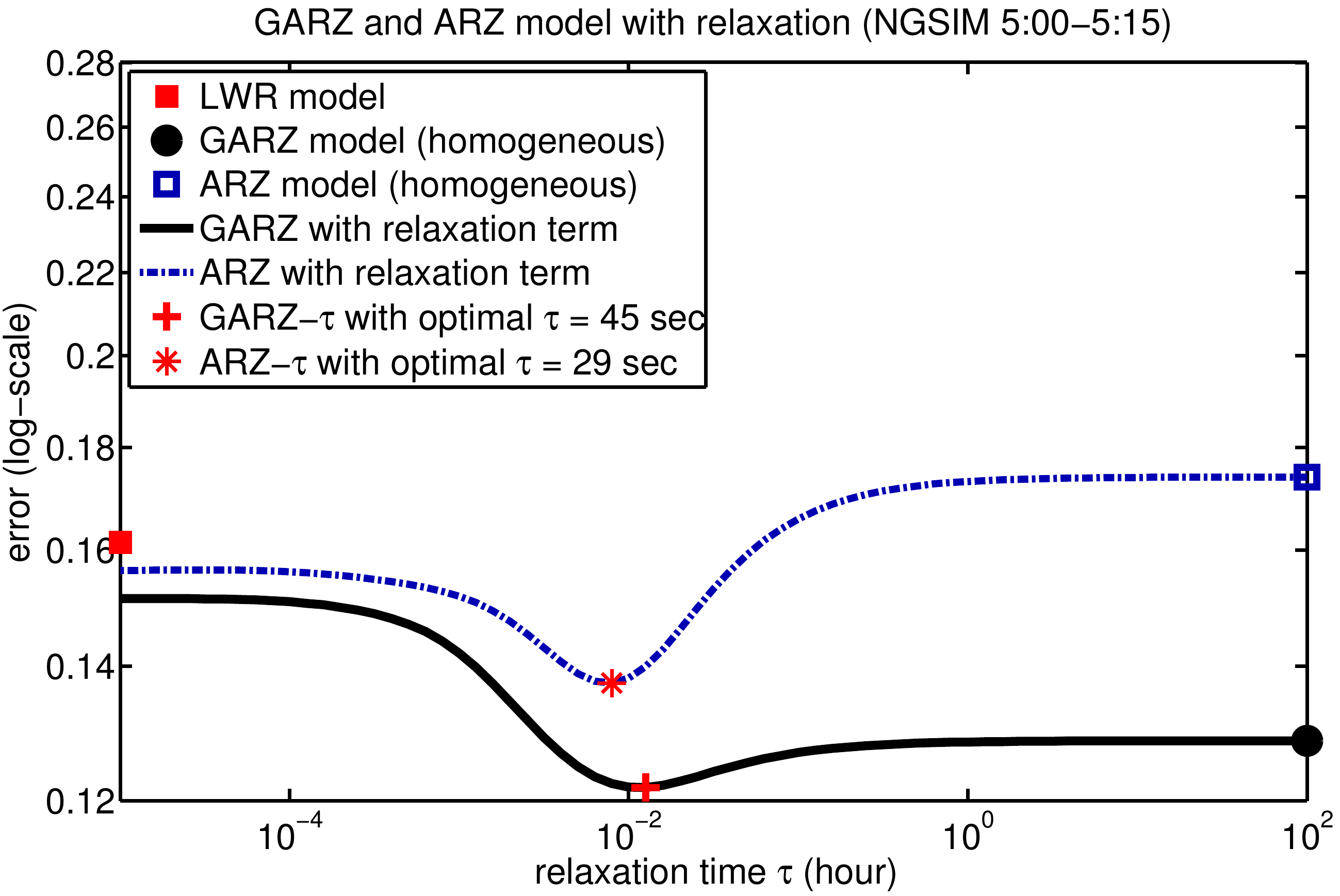}
\end{minipage}
\hfill
\begin{minipage}[b]{.49\textwidth}
\includegraphics[width=\textwidth]{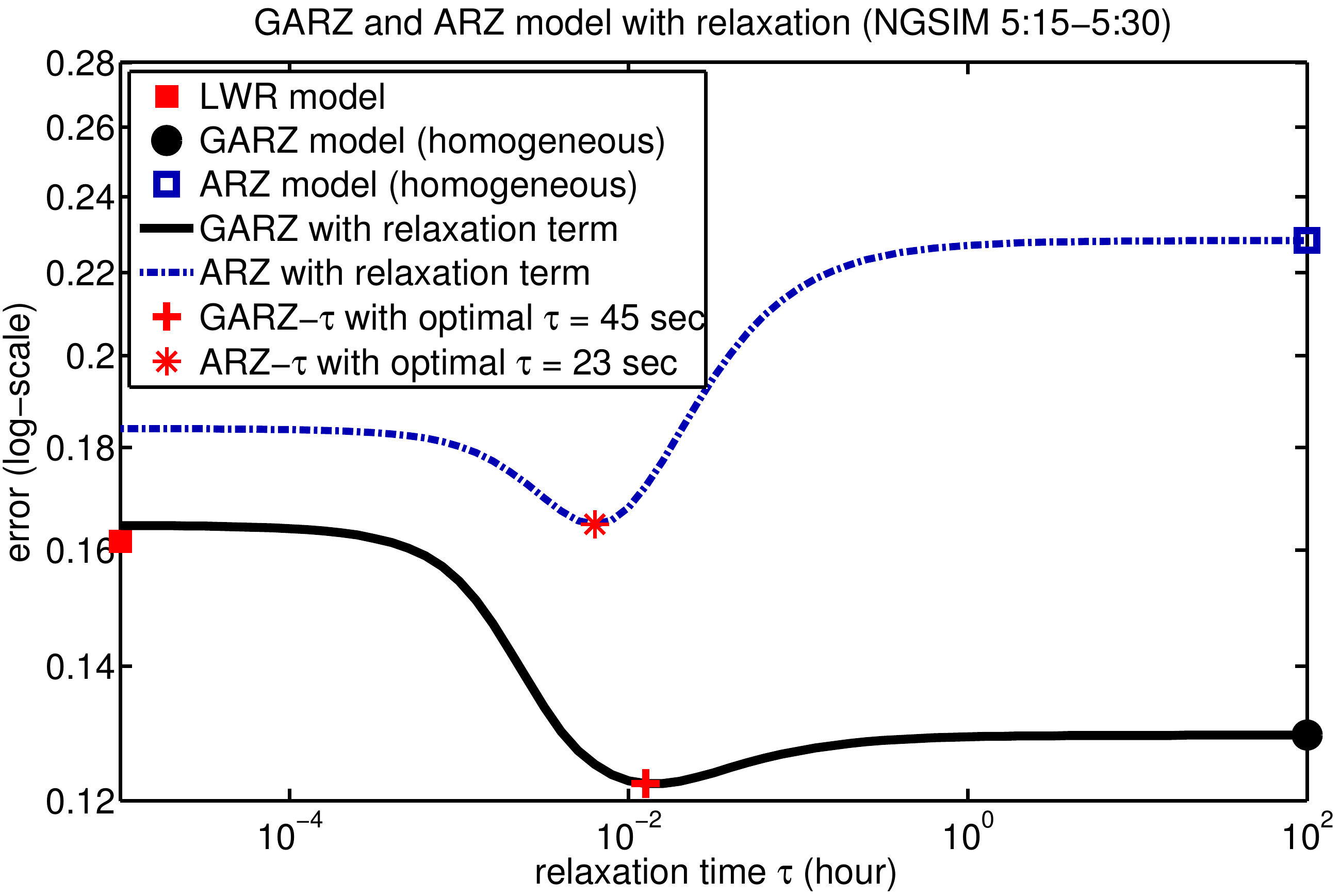}
\end{minipage}

\vspace{-.4em}
\caption{Model errors of the inhomogeneous ARZ and GARZ models, ARZ-$\tau$ (dashed blue) and GARZ-$\tau$ (solid black), as functions of the relaxation time $\tau$. Shown are the time-averaged errors \eqref{eq:error_t} for one day in the RTMC data set (top left), and spatio-temporally averaged errors \eqref{eq:error_xt} for the three NSGIM data sets: 4:00pm--4:15pm (top right), 5:00pm--5:15pm (bottom left), and 5:15pm--5:30pm (bottom right). Also shown are the errors with the homogeneous (i.e., $\tau\to\infty$) ARZ model (blue square) and GARZ model (black circle), the error obtained with the LWR model (red square), and the choices of $\tau$ that yield the smallest model error (red star and red plus).}
\label{fig:error_tau}
\end{figure}

The inhomogeneous second-order models are based on the same data-fitted functions as their homogeneous counterparts. The only new parameter is the relaxation time $\tau$. Since historic fundamental diagram data is commonly devoid of temporal information relevant for traffic dynamics, there is no canonical way to obtain the value of $\tau$ from historic data. Regarding realistic choices for $\tau$, the only information found in the literature is that it cannot be much smaller than 3~seconds due to physical restrictions of the vehicle engines. Due to this absence of a good modeling principle for the value of $\tau$, we conduct our model validation procedure for multiple instances of the inhomogeneous ARZ and GARZ models, where we let $\tau$ range from milliseconds to days. For each test, whichever choice of $\tau$ yields the smallest model error is then the optimal ARZ-$\tau$ (respectively optimal GARZ-$\tau$) model.

The described $\tau$-study is conducted for one day (January 8, 2003) in the RTMC data set, and for the three NGSIM data sets. For the RTMC data we evaluate the temporally averaged error \eqref{eq:error_t}, and for the NSGIM data the spatio-temporally averaged errors \eqref{eq:error_xt}. The results of the study are shown in Fig.~\ref{fig:error_tau}. In each test case, the ARZ-$\tau$ and GARZ-$\tau$ models are computed for many values of $\tau$ (blue and black functions, respectively). In addition, the (homogeneous) ARZ and GARZ models are computed (blue square and black circle). As expected, they agree with ARZ-$\tau$ and GARZ-$\tau$, respectively, for $\tau\gg 1$. Moreover, the LWR model is computed (red square). As argued above, due to the presence of boundary data with $u\neq U_\text{eq}(\rho)$, its results are different from ARZ-$\tau$ and GARZ-$\tau$ for $\tau\ll 1$. Finally, in each test case the particular $\tau = \tau_\text{opt}$ is marked (red star and red plus) for which ARZ-$\tau$ and GARZ-$\tau$, respectively, yield the smallest model errors. It is apparent that in all four cases shown in Fig.~\ref{fig:error_tau}, the inhomogeneous models possess an optimal relaxation time $0<\tau_\text{opt}<\infty$.

\begin{table}
\begin{tabular}{|ll||lr|lr|lr|lr|l}
\hline
Data set\rule{0em}{1.0em}\hspace{-1em} && ARZ\hspace{-1em}&&
ARZ-$\tau_{\text{opt}}$\hspace{-1em} & & GARZ\hspace{-.5em} & & GARZ-$\tau_\text{opt}$\hspace{-1em} &  \\
\hline\hline
NGSIM\rule{0em}{1.0em}
      & 4:00--4:15 & 0.153 &(+13\%) & 0.141 &  (+4\%) & 0.138 & (+2\%) & 0.135 & \\
NGSIM & 5:00--5:15 & 0.174 &(+43\%) & 0.137 & (+13\%) & 0.129 & (+6\%) & 0.122 & \\
NGSIM & 5:15--5:30 & 0.228 &(+87\%) & 0.165 & (+35\%) & 0.129 & (+6\%) & 0.122 & \\
\hline
RTMC\rule{0em}{0em}
     & congested   &0.208 &(+18\%) & 0.192 &(+8\%) & 0.184 &(+4\%) & 0.177 & \\
RTMC & non-cong.   &0.067 &(+6\%) & 0.066 &(+5\%) & 0.064 & (+2\%) & 0.063 & \\
\hline
\end{tabular}
\vspace{.6em}
\caption{Spatio-temporal average errors of the homogeneous ARZ and GARZ models, together with the inhomogeneous versions, ARZ-$\tau$ and GARZ-$\tau$, with optimal choices of the relaxation time $\tau$. Shown are the results for the three NGSIM data sets (rows 2--4), as well as the RTMC data, separated into congested and non-congested days (rows 5--6). In each row, the parentheses denote the excess error relative to the best model (GARZ-$\tau_\text{opt}$ in all cases). For the RTMC data, the optimal $\tau$ values are computed separately for each day; then, averages over all respective days are taken.}
\label{tab:model_errors_tau}
\end{table}

As one can see in Fig.~\ref{fig:error_tau}, the error-minimizing relaxation times are $\tau_\text{opt}\approx 25\text{s}$ for the ARZ model (the result for NGSIM 4:00pm--4:15pm is not as reliable, because the minimum is weakly pronounced), and $\tau_\text{opt}\approx 50\text{s}$ for the GARZ model in the NGSIM data. Moreover, for the RTMC data, the GARZ model yields $\tau_\text{opt}\approx 150\text{s}$. If the values of $\tau_\text{opt}$ give an indication about the time scales on which driving behavior evolves in reality, then their relatively large values give rise to the interesting observation that real driving behavior changes rather slowly; much slower than the time scales on which vehicles are able to accelerate due to engine power. In addition, one can observe that the GARZ model's optimal relaxation time is noticeably slower for the RTMC data than it is for the NGSIM data. While it is in principle possible that this difference stems from differences in driving behavior between Minnapolis vs.~the San Francisco Bay Area, more plausibly the difference stems from the combination of a relatively flat minimum of the black error curve, and from aggregation and lack of precision effects in the sensor data.

In order to quantify how much the inhomogeneous second-order models can improve the model accuracy relative to homogeneous second-order models, we consider the model errors of the ARZ-$\tau$ and GARZ-$\tau$ models (with $\tau = \tau_\text{opt}$ for each test case) in comparison with the errors of the ARZ and GARZ models. The resulting error values, together with the excess errors relative to the GARZ-$\tau_\text{opt}$ model shown in parentheses,  are shown in Table~\ref{tab:model_errors_tau}. One can see that the addition of a relaxation term can improve the accuracy of the ARZ model noticeably, in particular for flow at high traffic densities (see NSGIM 5:15--5:30). This is in part due to the fact that the ARZ model's unrealistic spread in flow rate curves for large densities is ameliorated by a relaxation towards the LWR curve. In contrast, the addition of a relaxation term (of the considered form) to the GARZ model does hurt the model accuracy, but it does not lead to significant improvements either; the GARZ model appears to be quite good already in its homogeneous form, at least for the data sets studied here.

\vspace{1.5em}
\section{Conclusions}
\label{sec:conclusions}
We have presented a systematic approach to construct a data-fitted generalized Aw-Rascle-Zhang (GARZ) model of traffic flow from historic fundamental diagram data. The modeling advantages and the mathematical properties of the GARZ model have been discussed. Moreover, the predictive accuracy of the GARZ model has been compared with other macroscopic traffic models via a three-detector test on vehicle trajectory and stationary sensor data. The actual model comparison has been carried out in a macroscopic sense, i.e., discretization effects are kept small, and thus the accuracy of the actual PDE is investigated.

The model comparison considers a hierarchy of models: (i) the first-order Lighthill-Whitham-Richards (LWR) model, that is defined via a single curve $Q(\rho)$ in the fundamental diagram; (ii) the second-order Aw-Rascle-Zhang (ARZ) model, that is defined via a family of curves in the fundamental diagram $Q_w(\rho) = Q(\rho)+\rho(w-Q'(0))$, each of which is the LWR curve plus a linear function; (iii) the second-order GARZ model, that is defined via a family of curves that are not simple transformations of each other, i.e., $Q_w(\rho) = \rho V(\rho,w)$, where $V(\rho,w)$ is a two-parameter generalized velocity function; and finally, the model validation is augmented by (iv) an Interpolation predictor that reconstructs the traffic density and velocity in the study domain via a simple linear interpolation of the boundary data. Moreover, the second-order traffic models are considered in their homogeneous form, ARZ and GARZ, as well as inhomogeneous versions thereof, ARZ-$\tau$ and GARZ-$\tau$, for which a relaxation term towards the LWR equilibrium state is added.

The general conclusions that can be drawn from the comparisons are as follows. First, and most importantly, the GARZ model yields the most accurate predictions, and it reproduces the behavior of the real data in the best fashion. Second, the ARZ model is superior to the LWR model for low traffic densities; but its model accuracy suffers noticeably as traffic becomes more congested. This observation can be interpreted as a manifestation of the ARZ model's multiple stagnation densities---a drawback that is addressed by the GARZ model. Third, in terms of averaged errors, the Interpolation predictor yields model errors that are similar to, or even slightly lower than, those of the LWR and ARZ models. At the same time, predictions based on mere interpolation are more oscillatory and less sharp than predictions obtained via traffic models. In addition, interpolation performs less well on longer road segments. Fourth, it is observed that the addition of a relaxation term can further improve the accuracy of second-order models. However, a noticeable improvement is only observed when the relaxation time $\tau$ is chosen well. In turn, when $\tau$ is selected too small, the inhomogeneous model could be less accurate than the corresponding homogeneous model.

The studies of the inhomogeneous second-order models reveal that the optimal relaxation times lie in the range 14--60 seconds (with some values even larger). This is two orders of magnitude larger than the drivers' reaction time, and one order of magnitude larger than what the vehicles' engine capabilities would allow. A possible explanation for such seemingly slow relaxation is that in the ARZ and the GARZ models, it is the drivers' empty road velocities that relax, i.e., their general driving behavior, and not the actual vehicle velocities.

One modeling shortcoming of the ARZ model that is not addressed by the GARZ model studied in this paper is that no unique functional relationship between flow rate and density for low densities (``free flow regime'') is allowed. To address this point, phase transition models (cf.~\cite{Colombo2003, Goatin2006}) and variations of the ARZ model that result from introducing an upper bound on the vehicle velocity \cite{ColomboMarcelliniRascle2003} have been proposed. Moreover, a specific phase transition model has been applied in a practical context \cite{BlandinWorkGoatinPiccoliBayen2011}. The possibility to allow for a unique flow rate vs.\ density relationship may be of importance in certain applications. As described in a companion paper \cite{FanPiccoliSeibold2014}, this task can also be achieved within the GARZ framework, by collapsing the fundamental diagram curves into a single function in the free flow regime.

\vspace{.5em}
\section*{Acknowledgments}
M. Herty was supported by Excellenz Cluster EXC128, DAAD 55866082, and BMBF KinOpt 05M2013. B. Seibold would like to acknowledge the support by the National Science Foundation. This work was supported through grant DMS--1318709, and partially supported through grants DMS--1115269 and DMS--1318641.

\bibliographystyle{plain}
\bibliography{references_complete}

\end{document}